\documentclass[pra,twocolumn,showpacs,showkeys]{revtex4-2}

\usepackage{bbm}
\usepackage{bm}
\usepackage{dcolumn}
\usepackage{enumitem}
\usepackage{graphicx}
\usepackage[caption=false]{subfig}
\usepackage[bookmarksnumbered=true]{hyperref} 
\hypersetup{
  colorlinks = true,
  linkcolor = blue,
  anchorcolor = blue,
  citecolor = blue,
  filecolor = blue,
  urlcolor = blue
}
\usepackage{mathtools}
\mathtoolsset{showonlyrefs}
\usepackage{microtype}
\usepackage{siunitx}
\usepackage{tikz}
\usetikzlibrary{arrows,automata,chains,matrix}
\usepackage{todonotes}
\usepackage{booktabs}
\usepackage{amssymb, amsthm}

\newtheorem{theorem}{Theorem}
\newtheorem{lemma}{Lemma}
\newtheorem{heuristic}{Heuristic}

\DeclarePairedDelimiter{\floor}{\lfloor}{\rfloor}
\DeclarePairedDelimiter{\ecnd}{\mathbb{E}_{t-1}[\,}{]\,}

\DeclareSIUnit{\debye}{Debye}

\usepackage{dsfont} 

\renewcommand{\d}[1]{\ensuremath{\operatorname{d}\!{#1}}}
\newcommand{\e}[1]{ {\mathrm{e}}^{ #1 } }
\newcommand{\im}{\mathrm{i}}

\newcommand{\expectation}[1]{ \mathbb{E} [ #1 ] }

\newcommand{\expectationBig}[1]{ \mathbb{E} \Bigl[ #1 \Bigr] }

\newcommand{\vect}[1]{ \boldsymbol{#1} }

\newcommand{\indicator}[1]{ \mathds{1} [ #1 ] }

\newcommand{\process}[2]{ \{ #1 \}_{ #2 } }

\newcommand{\pnorm}[2]{ \| #1 \|{}_{#2} }

\newcommand{\probability}[1]{ \mathbb{P} [ #1 ] }
\newcommand{\probabilityBig}[1]{ \mathbb{P} \Bigl[ #1 \Bigr] }

\newcommand{\naturalNumbersPlus}{ \mathbb{N}_{+} }
\newcommand{\naturalNumbersZero}{ \mathbb{N}_{0} }
\newcommand{\realNumbers}{ \mathbb{R} }

\newcommand{\iterand}[2]{ #1^{[#2]} }

\newcommand{\refFigure}[1]{{\textrm{Figure~\ref{#1}}}}
\newcommand{\refTable}[1]{{\textrm{Table~\ref{#1}}}}

\newcommand{\refTheorem}[1]{{\textrm{Theorem~\ref{#1}}}}

\newcommand{\refLemma}[1]{{\textrm{Lemma~\ref{#1}}}}

\newcommand{\refSection}[1]{{\textrm{Section~\ref{#1}}}}
\newcommand{\refAppendixSection}[1]{\textrm{Appendix~\ref{#1}}}

\def\eqcom#1{\overset{\textnormal{(#1)}}}

\newcommand{\itr}[2]{ \iterand{#1}{#2} }

\newcommand{\ket}[1]{ | #1 \rangle }

\def\({{\Bigl(}}
\def\){{\Bigr)}}

\newcommand{\ba}{\begin{array}}
\newcommand{\ea}{\end{array}}

\newcommand{\eps}{\varepsilon}

\newcommand{\xdeleted}[1]{\deleted{}} 

\allowdisplaybreaks
\def\jaron#1{\textcolor{red!75!black}{#1}}

\usepackage[acronym]{glossaries}
\makeglossaries

\newacronym{DRIG}{DRIG}{Decomposed Random Intersection Graph}
\newacronym{ERRG}{ERRG}{Erd\"{o}s--R\'{e}nyi Random Graph}
\newacronym{ODE}{ODE}{Ordinary Differential Equation}
\newacronym{RSA}{RSA}{Random Sequential Adsorption}
\newacronym{RIG}{RIG}{Random Intersection Graph}
\newacronym{RGG}{RGG}{Random Geometric Graph}
\newacronym[longplural={Quantum Mechanical Models of a Rydberg Gas}]{QMMRG}{QMMRG}{Quantum Mechanical Model of a Rydberg Gas}

\makeatletter
\newcommand*{\glsplainhyperlink}[2]{%
  \colorlet{currenttext}{.}
  \colorlet{currentlink}{\@linkcolor}
  \hypersetup{linkcolor=currenttext}
  \hyperlink{#1}{#2}%
  \hypersetup{linkcolor=currentlink}
}
\let\@glslink\glsplainhyperlink  
\makeatother

\begin{document}

\title{Modeling Rydberg Gases using Random Sequential Adsorption on Random Graphs}

\author{Daan Rutten}
\affiliation{%
H. Milton Stewart School of Industrial and Systems Engineering \\
Georgia Institute of Technology, Atlanta, USA
}
\author{Jaron Sanders}
\affiliation{%
Department of Mathematics \& Computer Science \\
Eindhoven University of Technology, PO Box 513, 5600 MB Eindhoven, The Netherlands
}

\date{\today}

\begin{abstract}
The statistics of strongly interacting, ultracold Rydberg gases are governed by the interplay of two factors: geometrical restrictions induced by blockade effects, and quantum mechanical effects. To shed light on their relative roles in the statistics of Rydberg gases, we compare three models in this paper: a quantum mechanical model describing the excitation dynamics within a Rydberg gas, a \gls{RSA} process on a \gls{RGG}, and a \gls{RSA} process on a \gls{DRIG}. The latter model is new, and refers to choosing a particular subgraph of a mixture of two other random graphs. Contrary to the former two models, it lends itself for a rigorous mathematical analysis; and it is built specifically to have particular structural properties of a \gls{RGG}. We establish for it a fluid limit describing the time-evolution of number of Rydberg atoms, and show numerically that the expression remains accurate across a wider range of particle densities than an earlier approach based on an \gls{RSA} process on an \gls{ERRG}. Finally, we also come up with a new heuristic using random graphs that gives a recursion to describe a normalized pair-correlation function of a Rydberg gas. Our results suggest that even without dissipation, on long time scales the statistics are affected most by the geometrical restrictions induced by blockade effects, while on short time scales the statistics are affected most by quantum mechanical effects. 
\end{abstract}



\maketitle

\section{Introduction}
\label{sec:Introduction}

Ultra cold gases with atoms in highly excited states---Rydberg atoms---have attracted substantial interest over recent years for their applications in quantum computing and the study of nonequilibrium phase transitions \cite{lukin_dipole_2001,jaksch_fast_2000,weimer_quantum_2008,comparat_dipole_2010,saffman_quantum_2010}. Rydberg atoms are significant because a powerful Van der Waals interaction causes the atoms to feel mutual interactions over mesoscopic distances (tens of micrometers); a single Rydberg atom in a gas causes energy level shifts that make it impossible for neighboring atoms to excite to the same state. This dipole blockade leads to intriguing networked and complex spatial behavior, resulting in effects linked to fundamental problems in condensed matter physics \cite{lukin_dipole_2001,liebisch_atom_2005,weimer_quantum_2008,viteau_cooperative_2012,hofmann_sub-poissonian_2013,malossi_full_2014,schauss_crystallization_2015, letscher2016bistability}. For an overview of applications, we refer the interested reader to \cite{journal_of_physics_B_rydberg_2016}. 

The computational effort to accurately simulate the quantum mechanical behavior of Rydberg gas grows exponentially in the number of atoms. Precise numerical calculations are therefore commonly limited to the order of ten atoms, although approximations can be used to raise this number \cite{robicheaux_many-body_2005,honing_steady-state_2013,petrosyan_dynamics_2013}. Theoretically, there has therefore been focus on models that are more tractable: for example, rate equations were proposed to describe the Rydberg gas \cite{ates_many-body_2007,ates_electromagnetically_2011,petrosyan2013spatial}; the Mandel Q parameter was studied under reversible dynamics \cite{ates_entropic_2012} and using \gls{RSA} processes on \glspl{ERRG} \cite{sanders2015sub}; and the physics of Rydberg atoms has e.g.\ been related to a nonlinear optical polarizability model \cite{wu2017quantum} as well as glassy soft-matter models \cite{perez2018glassy,gutierrez2019physical}. Also, ties to classical statistics can be found in related Ising models, dissipative approaches, and open quantum systems \cite{overbeck2017multicritical,foss2017solvable,macieszczak2020theory,pistorius2020quantum}.

The focus of this paper is on a new, generalized variant of the \gls{RSA} process described in \cite{sanders2015sub}, which is setup to mimic the Rydberg gas even more closely. Recall that the canonical, continuum two-dimensional \gls{RSA} process is that of randomly throwing disks of radius $r > 0$ one by one in a two-dimensional box, and in such a way that the disks do not overlap \cite{evans_random_1993,senger2000irreversible}. The spatial correlations in \gls{RSA} processes are notoriously challenging to analyze on such continuum, there being only an exception for the one-dimensional case (a line) \cite{renyi1958one}. To bypass this mathematical difficulty, one may opt to conduct \gls{RSA} on a random graph instead; and in fact, such approach was used to provide a closed-form expression for the Mandel Q parameter of a Rydberg gas as well as a description of the average number of Rydberg atoms over time \cite{sanders2015sub}. The comparison to experiments there shows that such model is capable of describing the statistics of a Rydberg gas accurately. 

This accuracy is surprising in light of the facts that \gls{RSA} and random graphs both belong to the realm of classical probability theory, and the vertices of e.g.\ the \gls{ERRG} \cite{erdos_random_1959} have no property corresponding to a physical position of a particle. This raises two fundamental questions. First, if we purposely construct a random graph that shares structural properties of the Rydberg gas, how effective can a \gls{RSA} process on it be at describing the excitation dynamics of the Rydberg gas? If there is a natural graph structure underlying the Rydberg atoms then the random graph should be designed to capture it. This is however challenging to achieve, because most random graphs do not have a geometrical structure (usually in order to remain mathematically tractable), and as we add structure we typically complicate the analysis. Second, which factor dominates the statistics of a Rydberg gas? When a Rydberg gas is carefully created in a laboratory so as to fully maintain its quantum mechanical behavior, the statistics of the Rydberg atoms will be influenced by the interplay of two factors: geometrical restrictions induced by blockade effects, and quantum mechanical effects. If geometrical restrictions primarily dictate the excitation dynamics of Rydberg gases, then even in a regime in which the Rydberg gas behaves completely quantum mechanically and without dissipation, a classical model capturing primarily the geometrical restrictions may still be effective in accurately describing the Rydberg gas' statistics.

We aim to shed light on these questions by comparing the statistics predicted by different models. Concretely, we implement numerically:
\begin{itemize}[topsep=2pt,itemsep=2pt,partopsep=2pt,parsep=2pt]
\item[A.] a \gls{QMMRG}; and
\item[B.] a \gls{RSA} process on a \gls{RGG}, 
\end{itemize}
in order to study the influence of quantum mechanical effects relative to the geometrical restrictions induced by blockade effects. We compare then the statistical predictions of models A and B to the statistical prediction of 
\begin{itemize}[topsep=2pt,itemsep=2pt,partopsep=2pt,parsep=2pt]
\item[C.] a \gls{RSA} process on a \gls{DRIG},
\end{itemize}
a new type of random graph purposely constructed to more accurately describe a \gls{RGG} than an \gls{ERRG} while simultaneously maintaining mathematical tractability. The idea here is that model C is constructed to mimic model B as closely as possible and to have mathematically analyzable features, while in turn model B is constructed to mimic model A.

Models A and B are challenging to mathematically analyze---it is hard, if not impossible, to derive an analytical expression for e.g.\ the ultimate number of Rydberg atoms. Model C on the other hand lends itself better for a rigorous analysis. For example, \refTheorem{thm:Jamming_limit_converges_in_probability} in this paper gives an implicit expression for the ultimate number of excited particles on a \gls{DRIG} when conducting \gls{RSA} on it. \emph{A fortiori,} \refTheorem{thm:Fluid_limit_of_Dt} proves that the fluid limit of \gls{RSA} on a \gls{DRIG} is given by a system of integral equations. Because model C mimics model B more accurately than an \gls{ERRG}, this automatically also provides an improved approximation (see \refFigure{fig:jamminglimit_error}) to the ultimate number of Rydberg atoms in model A when compared to the closed-form expression derived in \cite{sanders2015sub}.

Theorems~\ref{thm:Jamming_limit_converges_in_probability}--\ref{thm:Fluid_limit_of_Dt} also contribute to the mathematical perspective. The study of exploration processes of the \gls{RSA} type on random graphs has seen interest in recent years \cite{sanders2015sub,dhara2016generalized,bermolen2017jamming,bermolen2017scaling,dhara2018corrected,jonckheere2018asymptotic,meyfroyt2018degree,bermolen2020sequential}. The tools necessary from probability theory to analyze said processes include e.g.\ random graph couplings, fluid and diffusion limits, stochastic differential equations, and martingale analyses, and the current work extends this framework in two ways. First, we prove a fluid limit of an exploration process on \emph{a subgraph of a mixture} of two (random) graphs. This new approach requires us to determine a plethora of fluid limits, one for each of the relevant parts of the mixture of (random) graphs, see \refTheorem{thm:Fluid_limit_of_a_RIG_plus_isolated_vertices_exploration} and \refLemma{lem:Uniform_convergence_in_probability_of_the_scaled_processes} in its proof in \refAppendixSection{sec:Appendix__Uniform_convergence_in_probability}; and then combine these fluid limits carefully using a nonlinear time transformation that depends on an inverse exploration process, see the explanation above \refTheorem{thm:Fluid_limit_of_Dt} and its proof in \refAppendixSection{sec:Proof_of_time_rescaling_for_DRIG}. One further technical caveat of mixing exploration processes is that we lose a global Lipschitz' continuity property of the fluid limit. The canonical approach to prove a fluid limit, that is to utilize Lipschitz' continuity prior to an application of Gr\"{o}nwell's inequality, therefore breaks down. Our second contribution is the circumvention of this issue by conditioning only on sample paths that satisfy a local Lipschitz' continuity, and showing that these sample paths occur with probability one in a large graph limit: see \refAppendixSection{sec:Appendix__Uniform_convergence_in_probability}.

This paper is structured as follows. You have just read \refSection{sec:Introduction}, that is, our introduction. Next, \refSection{sec:Models} introduces models A, B, and C. Thereafter \refSection{sec:Results} describes a graph exploration algorithm that constructs a mixture of a so-called \gls{RIG} \emph{plus} a graph of isolated vertices (constituting a \gls{DRIG}), while simultaneously conducting \gls{RSA} on it. \refSection{sec:Results} then gives our theoretical results for model C, that is, Theorems~\ref{thm:Jamming_limit_converges_in_probability}--\ref{thm:Fluid_limit_of_Dt}. \refSection{sec:Heuristics} discusses a new heuristic using random graphs to describe a normalized pair correlation function of a Rydberg gas, and \refSection{sec:Numerics} numerically compares the statistics of models A, B, and C to each other as well as to the heuristic. Finally, \refSection{sec:Conclusion} concludes with a possibility for future research. The Appendices give rigorous mathematical proofs of Theorems~\ref{thm:Jamming_limit_converges_in_probability}--\ref{thm:Fluid_limit_of_Dt}.

\section{Models}
\label{sec:Models}

\glsreset{QMMRG}
\subsection{\texorpdfstring{\gls{QMMRG}}{QMMRG}}
\label{sec:quantum}

Consider an ensemble of $n \in \naturalNumbersPlus$ atoms that are distributed uniformly at random over a two-dimensional box $V = [a,b] \times [c,d] \subset \realNumbers^2$ with periodic boundary conditions (essentially, a torus). Each atom $i \in \{ 1, \ldots, n \}$ can be in either its ground level, $\sigma_i = 0$ say; or an excited level, $\sigma_i = 1$ say. The energy of level $l \in \{ 0, 1 \}$ will be denoted by $\epsilon_l$. A particular pure quantum mechanical state of the system is e.g.\ $\ket{ \sigma } = \ket{ \sigma_1 \cdots \sigma_n }$.

The atoms will undergo a dynamical process driven by a Hamiltonian $\mathcal{H}(t)$. We assume that the Hamiltonian is made up of three contributing terms, 
\begin{equation}
\mathcal{H}(t) 
= \mathcal{H}_a + \mathcal{H}_l(t) + \mathcal{H}_{\mathrm{vdW}},
\label{eqn:Hamiltonian}
\end{equation}
say. The first term $\mathcal{H}_a$ describes the total energy contribution of the particles being in the ground or excited level. For $\sigma \in \{0,1\}^n$, it satisfies
\begin{equation}
\langle \sigma \rvert \mathcal{H}_a \lvert \sigma \rangle 
= \sum_i \epsilon_{\sigma_i}.
\end{equation}
The second term $\mathcal{H}_l(t)$ describes the exposure of the atoms to a laser field. The laser facilitates the excitation process from the ground to the excited level, as well as the stimulated emission process from the excited to the ground level. Concretely, for any two states $\sigma, \zeta \in \{ 0, 1 \}^n$,
\begin{align}
&
\langle \sigma \rvert \mathcal{H}_l(t) \lvert \zeta \rangle 
\nonumber \\ &
= 
\begin{cases}
\frac{\hbar\Omega}{2} \exp{ \bigl( \im \omega t \sum_j ( \zeta_j - \sigma_j ) \bigr) } & \textnormal{if } \sum_j | \sigma_j - \zeta_j | = 1, \\
0 & \textnormal{otherwise}. \\
\end{cases}
\end{align}
Here, $\Omega$ denotes the Rabi frequency induced by the laser field, $\hbar$ is the reduced Planck constant, and $\omega$ represents the atom's natural frequency. Finally, the third term $\mathcal{H}_{\mathrm{vdW}}$ describes a strong Van der Waals interaction between any two atoms in the excited level. For any $\sigma \in \{ 0,1 \}^n$, we have 
\begin{equation}
\langle \sigma \rvert \mathcal{H}_{\mathrm{vdW}} \lvert \sigma \rangle = \sum_{i}\sum_{j \neq i} \indicator{ \sigma_i=\sigma_j=1 } \frac{C_6}{2\lVert r_i - r_j \rVert_2^6}.
\end{equation}
Here, $C_6$ denotes the Van der Waals coefficient. Whenever two excited atoms are close, the inverse relationship on the distance in the Van der Waals potential creates states with high energy resulting in a blockade effect. This geometric effect disallows atoms close to an excited atom to populate the excited level \cite{petrosyan2013spatial,petrosyan_dynamics_2013,ates_many-body_2007,robicheaux_many-body_2005}.

Together with Schr\"{o}dinger's equation, the Hamiltonian in \eqref{eqn:Hamiltonian} yields a system of differential equations 
\begin{equation}
\im \hbar \frac{ \partial }{ \partial t } \ket{ \Psi(t) } = \mathcal{H}(t) \ket{ \Psi(t) },
\quad
\ket{ \Psi(0) } = \ket{ 0 \cdots 0 }
\label{eqn:Schrodingers_equation}
\end{equation}
where 
$
\ket{\Psi(t)}
= \sum_{ \sigma \in \{0,1\}^n } c_{\sigma}(t) \ket{ \sigma },
$
which we will solve numerically for the complex coefficients $c_\sigma \in \mathbb{C}$, $\sigma \in \{0,1\}^n$. To tackle the exponential increase in computational effort, the state space will be truncated to a set of reasonably reachable states depending on the Van der Waals potential of the state. \emph{Viz.}, only states with a Van der Waals potential below a certain threshold have been included in the analysis. Our numerical verification in \refAppendixSection{app:simulation} indicates that applying such state truncation has no visible effect on the measured statistics.

\glsreset{RGG}
\glsreset{RSA}
\subsection{\texorpdfstring{\glspl{RGG}}{RGG}, and \texorpdfstring{\gls{RSA}}{RSA}}
\label{sec:randomgraph}

Consider an ensemble of $n \in \naturalNumbersPlus$ vertices that are distributed similarly uniformly at random over a two-dimensional box $V$ with periodic boundary conditions. Denote the random position of a vertex $u \in \{ 1, \ldots, n \}$ by $r_u \in V$. A \gls{RGG} is then constructed by drawing an edge between every pair of vertices that are within a distance of each other less than some fixed blockade radius $r_b \geq 0$. In other words, the \gls{RGG} $G^{\mathrm{RGG}} = ( \mathcal{V}, \mathcal{E} )$, where
\begin{equation}
\mathcal{V} = \{ 1, \ldots, n \}, 
\enskip
\textnormal{and}
\enskip
\mathcal{E} = \bigl\{ (u,v) \in \mathcal{V} \big\vert \pnorm{ r_u - r_v }{2} \leq r_b \bigr\}.
\end{equation}

We now consider the process of \gls{RSA} on such \gls{RGG} as a model for the excitation process. In \gls{RSA}, we assign each of the vertices one of the labels `\textsc{unaffected}, \textsc{excited}, \textsc{blocked}'. At iteration step $t = 0$, we label all of the vertices \textsc{unaffected}. At iteration step $t > 0$, we pick a vertex uniformly at random from the set of all vertices labeled \textsc{unaffected}. This vertex is then relabeled \textsc{excited}, and we relabel all of its \textsc{unaffected} neighbors \textsc{blocked}. The relabeling process continues until the random iteration step $T_n \in \{ 1, 2, \ldots, n \}$ at which all vertices are either labeled \textsc{excited} or \textsc{unaffected}. Borrowing terminology from \cite{sanders2015sub,bermolen2017jamming,bermolen2017scaling}, the resulting state is called a \emph{jamming limit}, and $T_n$ is a so-called \emph{hitting time}. Compared to the \gls{QMMRG}, we note that a vertex that is \textsc{unaffected} within the \gls{RSA} process corresponds to an atom that is still in a ground state because it has not \textsc{excited} yet, and moreover, this atom has not yet been \textsc{blocked} by another atom.

To ensure that the dynamics of $\ket{\Psi(t)}$ in the \gls{QMMRG} and the process of \gls{RSA} on a \gls{RGG} give a somewhat comparable heuristic, we opt to match model parameters. Concretely, we will look for some function $r_b = r_b(C_6)$ that gives similar behavior as follows. Consider the scenario in which there is no interaction in the \gls{QMMRG}, i.e., when $C_6 = 0$. Each atom in the \gls{QMMRG} will then undergo a Rabi oscillation, independently of all other atoms. In particular, the long-term mean fraction of excited particles satisfies
\begin{equation}
\lim_{s \to \infty} \frac{1}{s} \int_0^s \sum_{ \sigma \in \{0,1\}^n } | c_{\sigma}(t) |^2 \frac{ \sum_i \sigma_i }{n} \d{t}
= \tfrac{1}{2}.
\end{equation}
Suppose now that we set $r_b = 0$ in the \gls{RGG}. Each vertex would eventually be labeled \textsc{excited} in the \gls{RSA} process. Consequently the long-term mean fraction of excited particles would then equal one, in discrepancy with the \gls{QMMRG}. Supposing instead that $r_b > 0$, we would then indeed have that the long-term mean fraction of excited vertices is strictly less than one. We therefore opt for a model blockade radius of the form $r_b = \alpha + ( C_6 / \Omega )^{1/6}$ with $\alpha > 0$. The physics behind the second term are e.g.\ discussed in \cite{comparat_dipole_2010}. The parameter $\alpha$ should be chosen such that the long-term mean fraction of excited particles in the RGG is $\tfrac{1}{2}$ if $r_b = \alpha$. Recall that for a \gls{RGG} no closed-form expression is available in literature for this number; hence, we have numerically solved for it and found that $\alpha \approx 2.03$.

\glsreset{DRIG}
\subsection{\texorpdfstring{\glspl{DRIG}}{DRIGs}}
\label{sec:DRIG}

We now introduce the \gls{DRIG}. The idea is to decompose a \gls{RIG} and replace its set of isolated vertices by a set of isolated vertices of which we control the size precisely. The construction is depicted schematically in \refFigure{fig:drig}.

\begin{figure*}
\centering 
\subfloat[]{
  \includegraphics[width=0.24\textwidth]{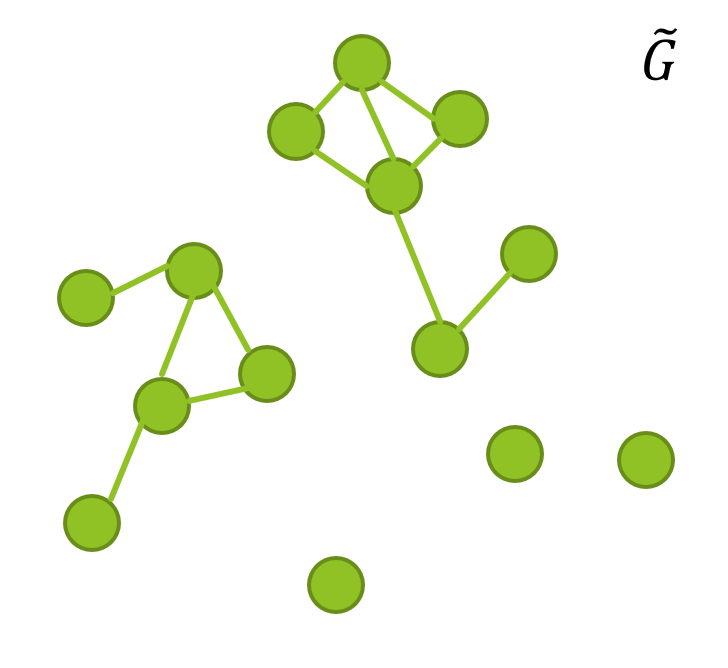}
}\qquad
\subfloat[]{
  \includegraphics[width=0.24\textwidth]{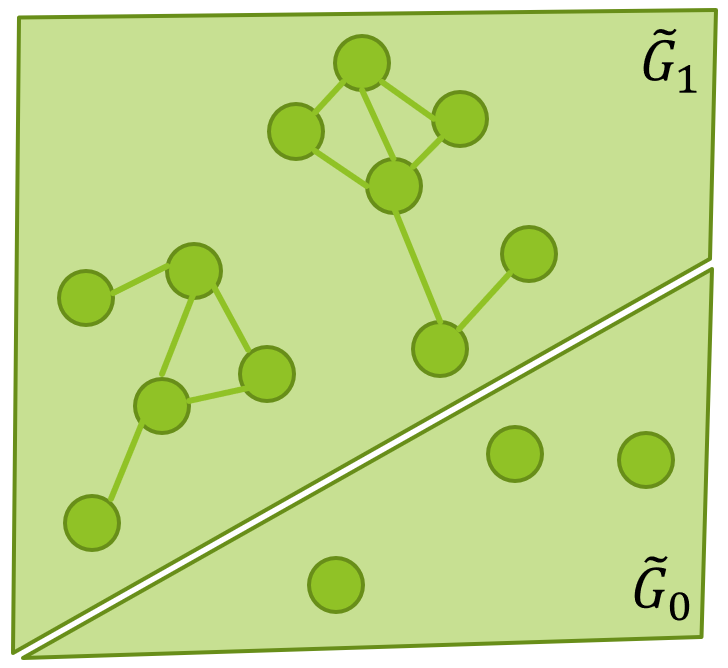}
}\qquad
\subfloat[]{
  \includegraphics[width=0.24\textwidth]{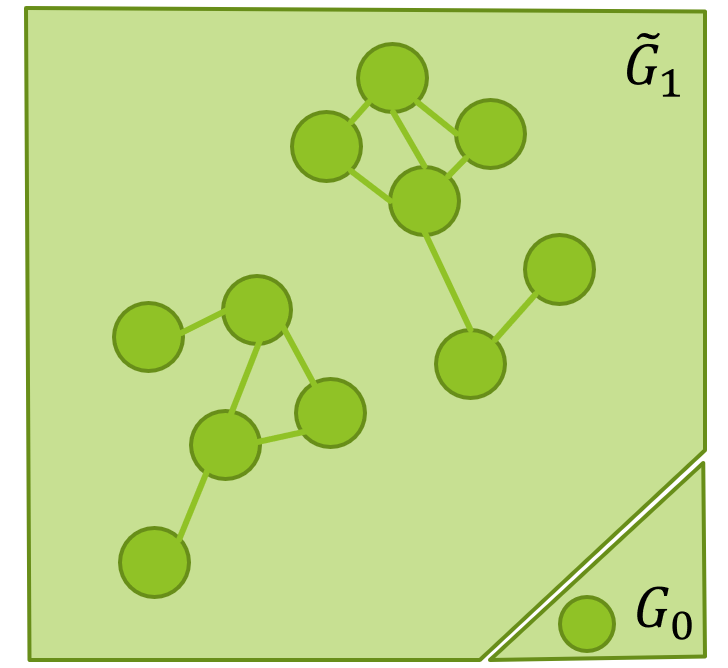}
}
\caption{Schematic depiction of the construction of a \gls{DRIG}. (a) First, a \gls{RIG} is constructed, which will usually be oversized. (b) Next, this \gls{RIG} is decomposed into the subgraph of isolated vertices and the remaining subgraph. (c) Finally, the subgraph of isolated vertices is removed and replaced by a graph of isolated vertices of specified size, which will usually be smaller.}
\label{fig:drig}
\end{figure*}

We will denote a \gls{RIG} by $G^{\mathrm{RIG}}( n^{\mathrm{RIG}}, \beta, \gamma )$; it is of size $n^{\mathrm{RIG}} \in \naturalNumbersPlus$ and has parameters $\beta, \gamma > 0$. It has a set of enumerable vertices, e.g.\ $\mathcal{V}^{\mathrm{RIG}} = \{ 1, \ldots, n^{\mathrm{RIG}} \}$, and we also associate a set of enumerable attributes $\mathcal{A}^{\mathrm{RIG}} = \{ 1, \ldots, \lfloor \beta n^{\mathrm{RIG}} \rfloor \}$ to it. To construct the \gls{RIG}, randomly connect every vertex $u \in \mathcal{V}^{\mathrm{RIG}}$ to each of the individual attributes $a \in \mathcal{A}^{\mathrm{RIG}}$ independently and with probability $\gamma / n^{\mathrm{RIG}}$. For $u \in \mathcal{V}^{\mathrm{RIG}}$, let $\mathcal{A}_v \subseteq \mathcal{A}$ be the set of attributes connected to vertex $u$. We now define the edge set of the \gls{RIG} as $\mathcal{E}^{\mathrm{RIG}} = \{ (u,v) \mid \mathcal{A}_u \cap \mathcal{A}_v \neq \varnothing \}$, i.e., the set of pairs of vertices that have at least one attribute in common.

We will denote a \gls{DRIG} by $G^{\mathrm{DRIG}}( n^{\mathrm{DRIG}}, \beta, \gamma, c)$, to indicate that it is of size $n^{\mathrm{DRIG}} \in \naturalNumbersPlus$ and has parameters $\beta, \gamma, c > 0$. Given values for these parameters, we construct the \gls{DRIG} as follows. First, set $n^{\mathrm{RIG}} = \sigma n^{\mathrm{DRIG}}$ where we introduce
\begin{equation}
\sigma 
= \frac{ 1 - e^{-c} }{ 1 - \xi_0 },
\quad 
\textnormal{and} 
\quad
\xi_0 
= \exp{ \bigl( -\beta \gamma (1 - e^{-\gamma}) \bigr) },
\end{equation}
and then construct $G^{\mathrm{RIG}}( n^{\mathrm{RIG}}, \beta, \gamma )$ as described above. Now decompose this \gls{RIG} into the subgraph $G^{\mathrm{RIG}}_0$ that consists of all of its isolated vertices, and the subgraph $G^{\mathrm{RIG}}_{1+}$ that consists of all of its vertices that have at least one neighbor. We now essentially remove the former subgraph and replace it: set $G^{\mathrm{DRIG}} = G^{\mathrm{DRIG}}_0 \cup G^{\mathrm{RIG}}_{1+}$, where $G^{\mathrm{DRIG}}_0$ is a graph consisting of $n_0^{\mathrm{DRIG}} = \lfloor e^{-c} n^{\mathrm{DRIG}} \rfloor$ isolated vertices.

The parameters of a \gls{DRIG} can be chosen to match three major features of the \gls{RGG}: (i) the mean number of neighbors of a vertex; (ii) the probability that a vertex is isolated; and (iii) the probability that a triangle occurs in the graph. For each of the four types of graphs that we have discussed so far, \refTable{tab:rg_comparison} gives the explicit expressions for each of these quantities. The ability to match feature (iii) is new compared to an earlier approach that uses an \gls{ERRG} \cite{sanders2015sub}. The \gls{DRIG} therefore captures more structural information of a \gls{RGG} than an \gls{ERRG}.

\begin{table}[hbtp]
\centering
\begin{tabular}{cccc}
\toprule
& $\mathbb{E}[N_v]$ & $\probability{ \{u,v\}, \{v,w\}, \{w,u\} \in \mathcal{E} }$ & $\probability{ N_v = 0 }$ \\
\midrule
\gls{RGG} & $c$ & $2 \int\limits_0^1 x I_{1 - \frac{x^2}{4}} \left( \frac{3}{2}, \frac{1}{2} \right) dx$ & $\e{-c}$ \\
\gls{ERRG} & $c$ & $0$ & $\e{-c}$ \\
\gls{RIG} & $\beta \gamma^2$ & $(1 + \beta\gamma)^{-1}$ & $\xi_0$ \\
\gls{DRIG} & $ \frac{1 - e^{-c}}{1 - \xi_0} \beta \gamma^2$ & $(1 + \beta\gamma)^{-1}$ & $\e{-c}$ \\
\bottomrule
\end{tabular}
\caption{Properties (i), (ii), and (iii) of a \gls{DRIG} compared to the other types of random graphs in the limit of $N \rightarrow \infty$. Here, $u \neq v \neq w$, $N_v$ denotes the number of neighbors of a vertex $v$ chosen uniformly at random, and $I_\cdot(\cdot,\cdot)$ the normalized incomplete beta integral.}
\label{tab:rg_comparison}
\end{table}

\section{Results}
\label{sec:Results}

\subsection{Time-dependent evolution of the for a \texorpdfstring{\gls{DRIG}}{DRIG} relevant quantities when conducting \texorpdfstring{\gls{RSA}}{RSA} on a \texorpdfstring{\gls{RIG}}{RIG}}
\label{sec:timedependent}

We now analyze the time-dependent behavior of a graph exploration algorithm that iteratively simultaneously constructs a mixture of a \gls{RIG} and a graph consisting solely of isolated vertices, as well as conduct \gls{RSA} on this graph. The \gls{RSA} process mimics the excitation process that happens within (a classical counterpart to) the \gls{QMMRG}. This idea of using a graph exploration algorithm to describe the excitation process in a Rydberg gas was first described in \cite{sanders2015sub}, and here we first generalize the approach for application to a mixture of a \gls{RIG} and a graph consisting solely of isolated vertices, and second lift the results through a time-transformation for application to a \gls{DRIG}.

\subsubsection{Graph exploration algorithm}
\label{sec:Graph_exploration_algorithm}

Suppose that parameters $n^{\mathrm{DRIG}} \in \naturalNumbersPlus$, $\beta, \gamma, c > 0$ are fixed and given. The algorithm will give each vertex one of the labels `\textsc{unaffected}, \textsc{excited}, \textsc{blocked}' iteratively as follows:

\emph{Initialization.} First, we label all vertices as being \textsc{unaffected}. The initial number of \textsc{unaffected} vertices in $G^{\mathrm{RIG}} = G_0^{\mathrm{RIG}} \cup G_{1+}^{\mathrm{RIG}}$ and $G^{\mathrm{DRIG}}_0$ are therefore 
\begin{equation}
U^{\mathrm{RIG}}(0) 
= n^{\mathrm{RIG}},
\enskip
U_{0}^{\mathrm{DRIG}}(0) 
= \lfloor \e{-c} n^{\mathrm{DRIG}} \rfloor
\end{equation}
say, respectively. The initial number of \textsc{excited} vertices in $G_{1+}^{\mathrm{RIG}}$ and $G_0^{\mathrm{DRIG}}$ are 
\begin{equation}
X_{1+}^{\mathrm{RIG}}(0)
= X_0^{\mathrm{DRIG}}(0) 
= 0
\end{equation}
say, respectively. We will define $n = n^{\mathrm{RIG}} + \lfloor \e{-c} n^{\mathrm{DRIG}} \rfloor$ here for notational convenience, which is the total number of vertices in the mixture of the \gls{RIG} and the graph consisting solely of isolated vertices that we are about to explore. 

\emph{Exploration.} Consider now iteration $t \in \naturalNumbersPlus$. For the $t$-th step of the exploration algorithm, select if possible a vertex $V(t)$ uniformly at random from the set of \textsc{unaffected} vertices, which has size 
\begin{equation}
U^{\mathrm{RIG}}(t-1) + U_0^{\mathrm{DRIG}}(t-1).
\end{equation}
Adjust the label of this vertex to \textsc{excited}. Change any label of an \textsc{unaffected} neighbor of $V(t)$ into \textsc{blocked}. 

\noindent
\emph{Case 1:} Consider now the event that the vertex $V(t)$ was picked from the graph $G^{\mathrm{RIG}}$. If this event occurs, then the number of \textsc{unaffected} vertices will have decreased by one plus its number of \textsc{unaffected} neighbors, i.e., 
\begin{equation}
U^{\mathrm{RIG}}(t) 
= U^{\mathrm{RIG}}(t-1) - 1 - N_{V(t)}(t-1),
\end{equation}
where $N_{V(t)}(t-1)$ denotes the number of \textsc{unaffected} neighbors of $V(t)$ after iteration $t-1$. Furthermore, 
\begin{equation}
U_0^{\mathrm{DRIG}}(t) = U_0^{\mathrm{DRIG}}(t-1).
\end{equation} 
The number of \textsc{excited} vertices \emph{within} $G^{\mathrm{RIG}}_{1+}$ only increases by one if $V(t)$ originated from $G^{\mathrm{RIG}}_{1+}$. In other words, 
\begin{gather}
X_{1+}^{\mathrm{RIG}}(t) = X_{1+}^{\mathrm{RIG}}(t-1) + \indicator{ N_{V(t)}(t-1) > 0 }, \nonumber \\ 
X_0^{\mathrm{DRIG}}(t) = X_0^{\mathrm{DRIG}}(t-1).
\end{gather}
Conditional on the history $\mathcal{F}(t)$, the event that $V(t)$ originates from $G^{\mathrm{RIG}}$ occurs with probability 
\begin{equation}
\frac{U^{\mathrm{RIG}}(t-1)}{U^{\mathrm{RIG}}(t-1) + U_0^{\mathrm{DRIG}}(t-1)}.
\end{equation}

\noindent
\emph{Case 2:} Consider now the complementary event: \emph{viz.}, the event that the vertex $V(t)$ came from $G^{\mathrm{DRIG}}_0$. If this event occurs, then the number of \textsc{unaffected} vertices decreases by one since this particular vertex has no neighbors, i.e., 
\begin{equation}
U^{\mathrm{RIG}}(t) 
= U^{\mathrm{RIG}}(t-1),
\enskip 
U_0^{\mathrm{DRIG}}(t) 
= U_0^{\mathrm{RIG}}(t) - 1.
\end{equation}
Moreover, the number of \textsc{excited} vertices increases by one, so 
\begin{equation}
X_{1+}^{\mathrm{RIG}}(t) 
= X_{1+}^{\mathrm{RIG}}(t-1),
\enskip
X_0^{\mathrm{DRIG}}(t) 
= X_0^{\mathrm{DRIG}}(t-1) + 1.
\end{equation} Conditional on the history $\mathcal{F}(t)$, this event occurs with probability 
\begin{equation}
\frac{U_0^{\mathrm{DRIG}}(t-1)}{U^{\mathrm{RIG}}(t-1) + U_0^{\mathrm{DRIG}}(t-1)}
\end{equation}

For the graph exploration algorithm just described, we are able to prove results on its \emph{jamming limit} and its \emph{fluid limit}.

\subsubsection{Jamming limit}

Define for $t \in \naturalNumbersZero$,
\begin{gather}
U(t) 
= U^{\mathrm{RIG}}(t) + U_0^{\mathrm{DRIG}}(t),
\quad
\textnormal{and}
\\
X(t) 
= X_0^{\mathrm{RIG}}(t) + X_{1+}^{\mathrm{RIG}}(t) + X_0^{\mathrm{DRIG}}(t).
\end{gather}
Consider now
\begin{equation}
T^* 
= \min \{ t \in \{ 0, 1, \ldots, n \} \mid U(t) = 0 \},
\end{equation}
the \emph{hitting time of zero} of the stochastic process $\process{ U(t)}{0 \leq t \leq n}$. Note that $T^*$ is a random variable and in particular that $T^*$ equals the number of \textsc{excited} vertices in the jamming limit. Leveraging this insight, we are able to prove convergence in probability of the fraction of excited vertices in the jamming limit to an implicitly characterized constant $\phi$ by conducting a hitting time analysis in \refAppendixSection{sec:Proof_that_the_jamming_limit_converges_in_probability}:

\begin{theorem}
\label{thm:Jamming_limit_converges_in_probability}
Let $T_n^* = T^* / n$. For any $\gamma > 0$, 
\begin{equation}
\probability{ | T_n^* - \phi | \geq \gamma } 
\to 0
\end{equation}
as $n \to \infty$. Here,
\begin{equation}
0 \leq \phi 
= \frac{e^{-c} + \sigma \arg\inf \bigl\{ 0 \leq s \leq 1 : \tilde{u}(s) = 0 \bigr \}}{e^{-c} + \sigma} \leq 1
\label{eqn:Inverse_expression_for_phi}
\end{equation}
where $\tilde{u}$ satisfies the following system of integral equations: for $0 \leq s \leq 1$,
\begin{align}
\tilde{u}(s)
&
= 1 - s - \int_0^s \gamma^2 \tilde{u}(y) \tilde{w}(y) \d{y},
\nonumber \\
\tilde{w}(s)
&
= \beta - \int_0^s \gamma \tilde{w}
(y) \d{y}.
\label{eqn:Simplified_ODE_for_phi}
\end{align}
\end{theorem}

Note that \refTheorem{thm:Jamming_limit_converges_in_probability} establishes that in this mixture of a \gls{RIG} and a graph consisting solely of isolated vertices, the final number of \textsc{excited} vertices satisfies
\begin{equation}
X(n) 
\approx \phi n.
\end{equation}
\refTheorem{thm:Jamming_limit_converges_in_probability} implies also that in a \gls{DRIG}, the final number of \textsc{excited} vertices 
\begin{align}
&
X^{\mathrm{DRIG}}(n^\mathrm{DRIG}) 
= X(n) - X_0^{\mathrm{RIG}}(n)
\nonumber \\ &
\approx \phi n - \sigma \xi_0 n^{\mathrm{DRIG}}
= (\phi( \sigma + \e{-c} ) - \sigma \xi_0) n^{\mathrm{DRIG}}
= \eta n^{\mathrm{DRIG}}
\end{align}
say. Moreover, \eqref{eqn:Inverse_expression_for_phi} together with \eqref{eqn:Simplified_ODE_for_phi} gives an implicit closed-form expression for $\phi$ and $\eta$.

\subsubsection{Fluid limits}

In fact, we are able to prove stronger fluid limits for the stochastic processes $\process{U(t)}{0 \leq t \leq n}$, $\process{X(t)}{0 \leq t \leq n}$ in \refTheorem{thm:Fluid_limit_of_a_RIG_plus_isolated_vertices_exploration}. Its proof is relegated to \refAppendixSection{sec:Proof_of_the_RIGplusOne_fluid_limit}. Fluid limits are functional laws of large numbers, and describe the function to which (an appropriately scaled variant of) the process converges to leading order. An example of the convergence of the RSA process is shown in \refFigure{fig:fluidlimit}. We require these fluid limits to ultimately describe the time-dependent behavior of the \gls{RSA} process on a \gls{DRIG} in \refTheorem{thm:Fluid_limit_of_Dt}.

\begin{figure}[hbtp]
\centering
\includegraphics[width=0.99\linewidth]{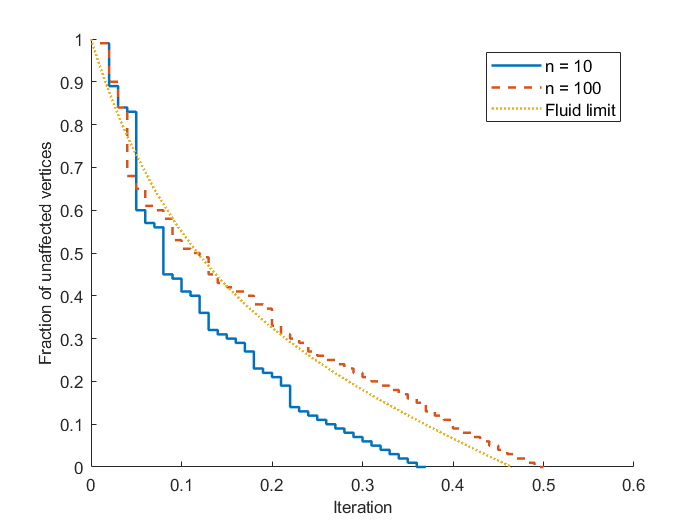}
\caption{The unaffected number of vertices of the \gls{RSA} process on the random graph for $n = 10$ and $n = 100$ vertices and the corresponding fluid limit $u(s) = u^{\mathrm{RIG}}(s) + u^{\mathrm{DRIG}}_0(s)$ in \refTheorem{thm:Fluid_limit_of_a_RIG_plus_isolated_vertices_exploration}.}
\label{fig:fluidlimit}
\end{figure}

\begin{widetext}
\begin{theorem}[Fluid limits of $U(t), X(t)$ when conducting \gls{RSA} on the mixture of a \gls{RIG} and isolated vertices]
\label{thm:Fluid_limit_of_a_RIG_plus_isolated_vertices_exploration}
For any $0 \leq S < \phi \leq 1$, $\eps > 0$,
\begin{equation}
\label{eq:unaffectedvertices}
\probabilityBig{ \sup_{s \in [0,S]} \Bigl| \frac{U( \lfloor s n \rfloor ) }{n} - (u^{\mathrm{RIG}} + u_0^{\mathrm{DRIG}})(s) \Bigr| > \eps } 
\to 0
\end{equation}
as $n \to \infty$. Furthermore, 
\begin{equation}
\label{eq:unaffectedvertices}
\probabilityBig{ \sup_{s \in [0,S]} \Bigl| \frac{X( \lfloor s n \rfloor ) }{n} - (x_0^{\mathrm{RIG}} + x_{1+}^{\mathrm{RIG}} + x_0^{\mathrm{DRIG}})(s) \Bigr| > \eps } 
\to 0
\end{equation}
as $n \to \infty$. Here, the functions $u^{\mathrm{RIG}}$, $u_0^{\mathrm{DRIG}}$, $x_0^{\mathrm{RIG}}$, $x_{1+}^{\mathrm{RIG}}$, and $x_0^{\mathrm{DRIG}}$ satisfy the following system of integral equations: for $0 \leq s \leq S < \phi \leq 1$ and $\tilde{\gamma} = (\sigma + e^{-c}) \gamma / \sigma$,
\begin{align}
u^{\mathrm{RIG}}(s) 
& = \frac{\sigma}{\sigma + e^{-c}} - \int\limits_0^s \frac{u^{\mathrm{RIG}}(y)}{u^{\mathrm{RIG}}(y) + u_0^{\mathrm{DRIG}}(y)} \bigl( 1 + \tilde{\gamma}^2 u^{\mathrm{RIG}}(y) w(y) \bigr) \d{y}, 
\\
u_0^{\mathrm{DRIG}}(s) 
& = \frac{e^{-c}}{\sigma + e^{-c}} - \int\limits_0^s \frac{u_0^{\mathrm{DRIG}}(y)}{u^{\mathrm{RIG}}(y) + u_0^{\mathrm{DRIG}}(y)} \d{y},
\\
w^{\mathrm{RIG}}(s) 
& = \frac{\beta \sigma}{\sigma + e^{-c}} - \int\limits_0^s \frac{u^{\mathrm{RIG}}(y)}{u^{\mathrm{RIG}}(y) + u_0^{\mathrm{DRIG}}(y)} \tilde{\gamma} w^{\mathrm{RIG}}(y) \d{y},
\\
x_0^{\mathrm{RIG}}(s) 
& = \int\limits_0^s \frac{u^{\mathrm{RIG}}(y)}{u^{\mathrm{RIG}}(y) + u_0^{\mathrm{DRIG}}(y)} \exp{ \Bigl( -\tilde{\gamma} w^{\mathrm{RIG}}(y) \bigl( 1 - e^{-\tilde{\gamma} ( \frac{\sigma}{\sigma + e^{-c}} -x_0^{\mathrm{RIG}}(y)-x_{1+}^{\mathrm{RIG}}(y) )} \bigr) \Bigr) } \d{y}, 
\\
x_{1+}^{\mathrm{RIG}}(s) 
& = \int\limits_0^s \frac{u^{\mathrm{RIG}}(y)}{u^{\mathrm{RIG}}(y) + u_0^{\mathrm{DRIG}}(y)} \Bigl( 1 - \exp{ \Bigl( -\tilde{\gamma} w^{\mathrm{RIG}}(y) \bigl( 1 - e^{-\tilde{\gamma} ( \frac{\sigma}{\sigma + e^{-c}}-x_0^{\mathrm{RIG}}(y)-x_{1+}^{\mathrm{RIG}}(y) ) } \bigr) \Bigr) } \Bigr) \d{y}, 
\\
x_0^{\mathrm{DRIG}}(s) 
& = \int\limits_0^s \frac{u_0^{\mathrm{DRIG}}(y)}{u^{\mathrm{RIG}}(y) + u_0^{\mathrm{DRIG}}(y)} \d{y}.
\label{eq:limits_processes}
\end{align}
\end{theorem}
\end{widetext}

\subsection{Time-dependent evolution of the fraction of excited vertices when conducting \texorpdfstring{\gls{RSA}}{RSA} on a \texorpdfstring{\gls{DRIG}}{DRIG}}
\label{sec:Time_dependent_evolution_of_the_fraction_of_excited_)vertices_when_conducting_RSA_on_a_DRIG}

The time-evolution of $U(t)$ for $t \in \{ 0, 1, \ldots, n \}$ does \underline{not} describe the number of \textsc{unaffected} vertices when conducting \gls{RSA} on a \gls{DRIG}, $D(t)$ for $t \in \{ 0, 1, \ldots, n^{\mathrm{DRIG}} \}$ say. Recall for instance that $n = n^{\mathrm{RIG}} + \lfloor \e{-c} n^{\mathrm{DRIG}} \rfloor$, which not necessarily equals $n^{\mathrm{DRIG}}$; the graph we are exploring is bigger. To obtain $D(t)$ one can utilize the graph exploration algorithm of \refSection{sec:Graph_exploration_algorithm}, but must then (i) filter out the vertices from $G^{\mathrm{RIG}}_0$, which are by definition not part of the \gls{DRIG}, and (ii) afterwards apply an appropriate nonlinear time transformation. The latter must be done carefully and is nontrivial because early in the graph exploration, relatively more vertices may have come from $G^{\mathrm{RIG}}_{1+}$. The issue of overcounting and rescaling of iteration time is also illustrated using an example sample path in \refTable{table:Explanation_of_filtering_out_vertices}. 

\begin{table*}[tbph]
\begin{tabular}{llp{5cm}}  
\toprule
Sequence
& Sample
& Remark
\\
\midrule
$\process{ U(t-1) - U(t) }{ t = 1, \ldots, n }$
& $2,1,2,3,1,5,1,4,2,1,1,1,3,4,1,6,2,1,\ldots, 0, 0, 0, \ldots, 0$ 
& Decrements describing the number of \textsc{unaffected} vertices in the \gls{RIG}. Identically zero after the random hitting time $T_n$.
\\
$\process{ \indicator{ V(t) \in \mathcal{V}(G_{1+}^{\mathrm{RIG}}) } }{ t = 1, \ldots, n }$
& $1,1,1,1,0,1,0,1,1,1,0,1,1,1,0,1,1,1,\ldots, 0, 0, 0, \ldots, 0$ 
& Steps at which a vertex was activated from $G_{1+}^{\mathrm{RIG}}$. These vertices are in the \gls{DRIG}.
\\
$\process{ \indicator{ V(t) \in \mathcal{V}(G_0^{\mathrm{DRIG}}) } }{ t = 1, \ldots, n }$
& $0,0,0,0,1,0,0,0,0,0,1,0,0,0,0,0,0,0,\ldots, 0, 0, 0, \ldots, 0$ 
& Steps at which a vertex was activated from $G_{0}^{\mathrm{DRIG}}$. These vertices are in the \gls{DRIG}. 
\\
$\process{ \indicator{ V(t) \in \mathcal{V}(G_0^{\mathrm{RIG}}) } }{ t = 1, \ldots, n }$
& $0,0,0,0,0,0,1,0,0,0,0,0,0,0,1,0,0,0,\ldots, 0, 0, 0, \ldots, 0$ 
& Steps at which a vertex was activated from $G_{0}^{\mathrm{RIG}}$. These vertices are \underline{not} in the \gls{DRIG}. 
\\
$\process{ D(t-1) - D(t) }{ t = 1, \ldots, n^{\mathrm{DRIG}} }$
& $2,1,2,3,1,5,\phantom{1,}\,4,2,1,1,1,3,4,\phantom{1,}\,6,2,1, \ldots, 0, 0, 0, \ldots, 0$ 
& Decrements describing the number of \textsc{unaffected} vertices in the \gls{DRIG}. For sufficiently large $n$, this is a shorter sequence with high probability.
\\
\bottomrule
\end{tabular}
\caption{A sample sequence of decrements that describes the number of \textsc{unaffected} vertices in the mixture of a \gls{RIG} with isolated vertices. We show how from it, we can find a coupled sequence of decrements describing the number of \textsc{unaffected} vertices in a \gls{DRIG}. This illustrates the issue of overcounting and necessity of transforming iteration time when going from a result for the mixture of the \gls{RIG} with isolated vertices to a result for the \gls{DRIG}.}
\label{table:Explanation_of_filtering_out_vertices}
\end{table*}

Concretely, the filtering out of vertices can be done by setting for $t \in \{ 0, 1, \ldots, n \}$,
\begin{align}
&
U^{\mathrm{DRIG}}(t) 
= U_0^{\mathrm{DRIG}}(t) + U_{1+}^{\mathrm{DRIG}}(t)
\nonumber \\ &
= U_0^{\mathrm{DRIG}}(t) + U^{\mathrm{RIG}}(t) + X_0^{\mathrm{RIG}}(t) - X_0^{\mathrm{RIG}}(n).
\end{align}
Rescaling iteration time can then be done appropriately by levering the inverse process 
\begin{equation}
( X^{\mathrm{DRIG}} )^\gets(u) 
= \inf{ \bigl\{ t \in \naturalNumbersZero \big| X_0^{\mathrm{DRIG}}(t) + X_{1+}^{\mathrm{RIG}}(t) = u \bigr\} }
\end{equation}
for $u \in \{ 0, 1, \ldots, n^{\mathrm{DRIG}} \}$ to only count those steps that involve excitations from the \gls{DRIG}. That is, for $t \in \{ 0, 1, \ldots, n^{\mathrm{DRIG}} \}$, we define
\begin{equation}
D^{\mathrm{DRIG}}(t) 
= U^{\mathrm{DRIG}}( ( X^{\mathrm{DRIG}} )^\gets(t) ).
\label{eqn:Unaffected_particles_in_the_DRIG}
\end{equation}
Fortunately, all fluid limits necessary for obtaining a fluid limit for \eqref{eqn:Unaffected_particles_in_the_DRIG} are provided by \refTheorem{thm:Fluid_limit_of_a_RIG_plus_isolated_vertices_exploration}. We prove the following in \refAppendixSection{sec:Proof_of_time_rescaling_for_DRIG}:

\begin{theorem}[Fluid limit of $D(t)$ when conducting \gls{RSA} on a \gls{DRIG}]
\label{thm:Fluid_limit_of_Dt}
For $0 \leq S < \eta \leq 1$, $\eps > 0$,
\begin{equation}  
\probabilityBig{ \sup_{s \in [0,S]} \Bigl| \frac{D( \lfloor s n^{\mathrm{DRIG}} \rfloor )}{n^{\mathrm{DRIG}}} - d(s) \Bigr| \geq \eps }
\to 0 
\end{equation}
as $n^{\mathrm{DRIG}} \to \infty$. Here, for $0 \leq s \leq S < \eta \leq 1$,
\begin{align}
&
d(s) 
= 
(\sigma + e^{-c}) u^{\mathrm{DRIG}}\bigl((x^{\mathrm{DRIG}})^\gets((\sigma + e^{-c})^{-1} s)\bigr),
\end{align}
and
\begin{align*}
u^{\mathrm{DRIG}}(s) &= (u^{\mathrm{RIG}} + u_0^{\mathrm{DRIG}} + x_0^{\mathrm{RIG}})(s) - \xi_0, \\
x^{\mathrm{DRIG}}(s) &= (x_{1+}^{\mathrm{RIG}} + x_0^{\mathrm{DRIG}})(s).
\end{align*}
\end{theorem}

\subsection{Translating from iteration steps to actual time}

Finally, if one is interested in translating iteration steps (in the arb.\ unit) to actual time (in e.g.\ seconds), then one may proceed as follows. Assume that the time between two steps is exponentially distributed with parameter $\lambda D(m) > 0 \, [\si{\per\second}]$ and independent of the past. \emph{Viz.}, $T_{m+1} - T_m \sim \textnormal{Exp}(\lambda D(m)) \, \si{\second}$ for all $m \in \naturalNumbersPlus$. Define then for $t \geq 0 \, [\si{\second}]$, the process $X'(t) = X( \sup \{ m \in \naturalNumbersZero \mid t \geq T_m \} )$ whose domain is actual time in seconds. This process satisfies for all $t \geq 0 \, [\si{\second}]$,
\begin{equation}
\frac{ \expectation{X'(t)} }{n^{\mathrm{DRIG}}} \, [\mathrm{arb.}]
\to \int_0^t \lambda d(x(s)) \d{s} \, [\mathrm{arb.}]
\end{equation}
as $n^{\mathrm{DRIG}} \to \infty$. This was also discussed in \cite{sanders2015sub}.

\section{Heuristics}
\label{sec:Heuristics}

In this section, we come up with a heuristic using random graphs that gives a recursion to describe the following \emph{normalized pair-correlation function}
\begin{equation}
h_\eps(r)
= \frac{ \mathbb{P}_N[ v, w \in \hat{X}(\infty) \mid d(v, w) \in (r - \epsilon, r + \epsilon) ] }{  \mathbb{P}_N[ v \in \hat{X}(\infty) ]^2 }
\label{eq:paircorrelation}
\end{equation}
say, for a Rydberg gas. Here, $\hat{X}(\infty)$ denotes the set of excited vertices in the jamming limit.  We chose this function for comparison purposes to \cite[Fig.~2]{robicheaux2005many}.

Many random graphs are however not built around an explicit notion of distance $d(v,w)$ between two vertices $v, w \in \mathcal{V}$, so evaluation of \eqref{eq:paircorrelation}'s numerator is impossible. Given a random graph without a notion of distance, such as an \gls{ERRG}, we need to instead use a proxy for the distance between two vertices. Using such a proxy comes at a cost: the following approach is a heuristic, and we have no mathematical results guaranteeing its accuracy. Nonetheless, the approaches will turn out descriptive. As our proxy for a notion of distance, we will specifically use the \emph{number of neighbors that two vertices have in common}. The idea is that if two vertices have more neighbors in common, then their distance should be smaller.

\begin{widetext}
\begin{heuristic}[Recursive system to approximate a normalized pair correlation function of a \gls{RGG} or the \gls{QMMRG}]
    \label{heuristic:Pair_correlation_function}
    For all $v, w \in \mathcal{V}$, $r \in (0, 2 r_b]$, $0 < \epsilon \ll 1$: the numerator in \eqref{eq:paircorrelation} satisfies
    \begin{align}
        \mathbb{P}_N[ v, w \in \hat{X}(\infty) \mid d(v, w) \in (r - \epsilon, r + \epsilon) ]
        &
        =
        \begin{cases}
            0 
            & 
            \textnormal{if } 0 < r < r_b 
            \\
            \displaystyle\sum_{ i,k } A_{i,k}^N(r) B_{i,k}^N(r) C_i^N(r)
            \approx \displaystyle\sum_{ i,k } \hat{A}_{i,k}^N(r) \hat{B}_{i,k}(r) \hat{C}_i^N
            & 
            \textnormal{otherwise}
            \\
        \end{cases}
        \label{eqn:Heuristic_used_to_calculate_a_pair_correlation_function}      
    \end{align}
    Here,
    \begin{align}
        A_{i,k}^N(r)
        &
        = \mathbb{P}_N
        [ 
            v, w \in \hat{X}(\infty) 
            \mid 
            \lvert \hat{N}_v \cap \hat{N}_w \rvert = k, 
            N_v = i, 
            d(v, w) \in (r - \epsilon, r + \epsilon) 
        ],
        \nonumber \\
        B_{i,k}^N(r)
         &
        = \mathbb{P}_N
        [ 
            \lvert \hat{N}_v \cap \hat{N}_w \rvert = k 
            \mid 
            N_v = i, 
            d(v, w) \in (r-\epsilon,r+\epsilon) 
        ],
        \nonumber \\
        C_i^N(r)
         &
        = \mathbb{P}_N
        [ 
            N_v = i 
            \mid 
            d(v,w) \in (r-\eps,r+\eps) 
        ],
        \label{eqn:Terms_in_our_application_of_the_law_of_total_probability}
    \end{align}
    and
        \begin{align}
            A_{i,k}^N(r)
            &
            \approx \hat{A}_{i,k}^N(r)
            \nonumber \\ &
            = 
            \frac{1}{N} 
            \hat{E}^{N-i-1}(r)
            + 
            \frac{1}{N} 
            \sum_{j=0}^{N-i-2} 
            \binom{N-i-2}{j} \bigl( p ( 1 - F_1(r) ) \bigr)^j \bigl( 1 - p ( 1 - F_1(r) ) \bigr)^{N-i-2-j} 
            \hat{D}_{i-k}^{N-k-j-1}
            \label{eqn:Explicit_expressions_for_the_terms_in_the_heuristic}
            \\ &
            \phantom{=}
            + \frac{N-i-2}{N} \bigl( 1 - p ( 1 - F_1(r) ) \bigr) 
            \nonumber \\ &
            \phantom{= +} 
            \times 
            \sum_{j=0}^{N-i-3}
            \sum_{l=0}^k
            \sum_{m=0}^{i-k}
            \binom{N-i-3}{j} ( p F_2(r) )^j \bigl( 1 - p F_2(r) \bigr)^{N-i-3-j}
            \binom{k}{l} p^l (1-p)^{k-l}
            \nonumber \\ &
            \phantom{ 
                = \times
                \sum_{j=0}^{N-i-3}
                \sum_{l=0}^k
                \sum_{m=0}^{i-k}
            }
            \times             
            \binom{i-k}{m} \bigl( p (1-F_1(r)) \bigr)^m \bigl( 1 - p (1-F_1(r)) \bigr)^{i-k-m}
            \hat{A}_{i-l-m,k-l}^{N-j-l-m-1}(r)
        \end{align}
    as well as
    \begin{gather}
        B_{i,k}^N(r)
        \approx \hat{B}_{i,k}(r)
        = \binom{i}{k} ( F_1(r) )^k \bigl( 1 - F_1(r) \bigr)^{i-k},
        \label{eqn:Heuristic_term_for_Bik}
        \\
        C_i^N(r)
        \approx \hat{C}_i^N
        = \binom{N-2}{i} p^i (1-p)^{N-2-i}.
        \label{eqn:Heuristic_term_for_Ci}
    \end{gather}
    Here, $p = c /N$,
        \begin{align}
            \hat{D}_i^N(r)
            &= 
            \frac{1}{N}
            + \frac{N-i-1}{N} 
            \sum_{k=0}^{N-i-2} 
            \sum_{l=0}^i 
            \binom{N-i-2}{k} p^k (1-p)^{N-i-2-k} 
            \nonumber \\
            &\phantom{
                = 
            	\frac{1}{N}
        	    + \frac{N-i-1}{N} 
    	        \sum_{k=0}^{N-i-2} 
	            \sum_{l=0}^i
            }
            \times \binom{i}{l} \bigl( p ( 1 - F_1(r) ) \bigr)^l \bigl( 1 - p ( 1 - F_1(r) ) \bigr)^{i-l}
            \hat{D}_{i-l}^{N-k-l-1}(r)
            \label{eqn:Heuristic_term_for_PnvinXinfty}            
            \\
            \hat{E}^N(r)
            &= 
            \frac{1}{N} 
            + \frac{N-1}{N} \bigl( 1 - p ( 1 - F_1(r) ) \bigr)
            \sum_{k=0}^{N-2} \binom{n-2}{k} p^k (1-p)^{N-2-k} 
            \hat{E}^{N-k-1}(r);
            \label{eqn:Heuristic_term_for_PnvinXinfty2}       
        \end{align}
    and finally
    $
        F_1(r)
        = \frac{1}{\pi} \bigl( 2 \arccos{ \frac{r}{2r_b} } - \frac{r}{2r_b} \sqrt{4 - ( \frac{r}{r_b} )^2 } \bigr)
    $
    and
    $
        F_2(r) = \frac{\textsc{AREA}(C_1 \cap C_2 \cap C_3)}{\textsc{AREA}(C_1)}$, where $C_1$, $C_2$ and $C_3$ are circles with radii $r_b$, the distance between $C_1$ and $C_2$ is $r$, the distance between $C_1$ and $C_3$ is $r_b$ and the distance between $C_2$ and $C_3$ is $r_b$ (see \cite{fewell2006area} for a closed-form expression).
\end{heuristic}
\end{widetext}

To understand how Heuristic~\ref{heuristic:Pair_correlation_function} is established, start by noting that the first equality in \eqref{eqn:Heuristic_used_to_calculate_a_pair_correlation_function} follows from the law of total probability and is still exact; essentially, we have partitioned the sample space according to the conditional events implicitly defined in \eqref{eqn:Terms_in_our_application_of_the_law_of_total_probability}. Next, we approximate each of these terms by relating them to a mathematical expression obtained from a random graph. For each term we use a random graph for which we believe the exploration process behaves similar to a classical counterpart of the excitation process of a \gls{QMMRG}.

\noindent
\emph{1.}\ For the terms $A_{i,k}^N(r)$ and $C_i^N(r)$, we use what would be their explicit recursions in a random graph with \emph{one geometric feature} and \emph{four classes of distances}, that is, the expressions for $\hat{A}_{i,k}(r)$ and $\hat{C}_i^N$ in \eqref{eqn:Explicit_expressions_for_the_terms_in_the_heuristic} and \eqref{eqn:Heuristic_term_for_Ci}. Concretely, in this random graph: 
\begin{itemize}[noitemsep,wide,labelindent=0pt]
    \item[a.] none of the vertices have an explicit position;
    \item[b.] vertices $v$ and $w$ are considered to be at exactly a distance $r$ of each other;
    \item[c.] all other vertices are either
    \begin{itemize}[noitemsep,wide,labelindent=0pt]
        \item[$\quad$ class 1] -- within a distance $r_b$ of $v$ but not $w$; or
        \item[$\quad$ class 2] -- within a distance $r_b$ of $w$ but not $v$; or
        \item[$\quad$ class 3] -- within a distance $r_b$ of both $v,w$; or
        \item[$\quad$ class 4] -- outside of a distance $r_b$ of both $v,w$.
        \end{itemize}
\end{itemize}
Depending on which class a vertex belongs to, the vertex's probability of creating an edge with vertices of its own class and of other classes changes. For example (ignoring the precise conditioning for a moment), once $w$ is explored before $v$, then from that point onward, class 1 vertices will connect to $v$ with reduced probability $p(1-F_1(r))$ instead of the originally higher probability $p$. Here, the function $F_1(r)$ describes the fraction of surface area of one circle of radius $r_b$ that overlaps with the surface area of a second circle also of radius $r_b$ and that is a distance $r$ away from the first circle. As a consequence, $r$ can be at most $2 r_b$ for the heuristic to work, otherwise the surface areas of the two circles would not overlap. A similar argument is used for all other cases, and for each case a suitable edge probability is chosen. It must be noted that these choices do not correspond exactly to the edge probabilities in e.g.\ a \gls{RGG}---this is in fact impossible because there is almost no geometry in the present random graph---and are chosen instead for them being reasonable proxies (usually, they are upper bounds). The recursions are then found by conditioning on the first excitation. 

\noindent
\emph{2.}\ For the term $B_{i,k}^N(r)$ we use what would be its explicit expressions in a \gls{RGG}, \emph{viz}., the expression $\hat{B}_{i,k}(r)$ in \eqref{eqn:Heuristic_term_for_Bik}. 

The resulting formula in Heuristic~\ref{heuristic:Pair_correlation_function} is not particularly attractive. However, the right-hand side of \eqref{eqn:Heuristic_used_to_calculate_a_pair_correlation_function} combined with the expressions in \eqref{eqn:Terms_in_our_application_of_the_law_of_total_probability}--\eqref{eqn:Heuristic_term_for_PnvinXinfty} do constitute a recursive system of expressions that can be evaluated. In particular, have a look at \refFigure{fig:paircorrelation} in \refSection{sec:comparison} as well as the surrounding discussion.

\section{Numerics}
\label{sec:Numerics}

\subsection{Comparison: \texorpdfstring{\gls{RSA}}{RSA} on (a \texorpdfstring{\gls{DRIG}}{DRIG} vs an \texorpdfstring{\gls{ERRG}}{ERRG})}

For a range of atomic densities, \refFigure{fig:jamminglimit_error} shows the relative difference between the ultimate mean number of excited atoms as predicted by: 1.\ \gls{RSA} on an \gls{ERRG}, 2.\ \gls{RSA} on a \gls{DRIG}, and 3.\ \gls{RSA} on a \gls{RGG}. Note that a closed-form expression is available in literature for this number in model 1 \cite{sanders2015sub}; that we have derived an implicit expression for this number in model 2 in \refTheorem{thm:Jamming_limit_converges_in_probability} in this paper; and that no analytical expression is available in literature for this number in model 3. Hence, the results for model 3 displayed here were obtained via numerical simulation only. Observe that the predictive power of model 2 for model 3 is improved across a wider range of atomic densities when compared to the predictive power of model 1 for model 3. Since model 3 mimics the excitation process in the \gls{QMMRG}, we conclude that a \gls{RSA} process on a \gls{DRIG} allows for more accurate prediction of the statistics of the excitation dynamics in a Rydberg gas, than what was achieved by \gls{RSA} on an \gls{ERRG} \cite{sanders2015sub}. From the mathematical perspective, it is also noteworthy that this observation indicates that the graph structure of a \gls{DRIG} matches the graph structure of a \gls{RGG} better. 

\begin{figure}[hbtp]
\centering
\includegraphics[width=0.99\linewidth]{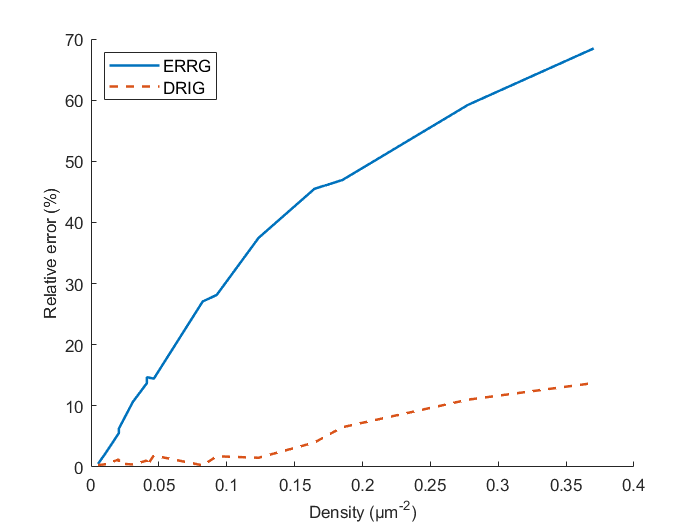}
\caption{The relative difference between the mean ultimate number of excited atoms as predicted by \gls{RSA} on an \gls{ERRG}, a \gls{DRIG}, and a \gls{RGG}; as a function of atomic density.}
\label{fig:jamminglimit_error}
\end{figure}

\subsection{Comparison:\texorpdfstring{\\}{} \texorpdfstring{\gls{QMMRG}}{QMMRG} vs \texorpdfstring{\gls{RSA}}{RSA} on a (\texorpdfstring{\gls{RGG}}{RGG} vs \texorpdfstring{\gls{DRIG}}{DRIG})}
\label{sec:comparison}

For two different atomic densities, \refFigure{fig:Mean_fraction_of_atoms_as_a_function_of_time_for_models_A_B_C} shows the mean number of excited atoms as a function of time as predicted by: A.\ the \gls{QMMRG}, B.\ \gls{RSA} on a \gls{RGG}, and C.\ \gls{RSA} on a \gls{DRIG}. These curves were obtained numerically for models A and B through simulation, while for model C we implemented \refTheorem{thm:Fluid_limit_of_Dt}. Recall that, as discussed in \refSection{sec:Models}, these three models have been matched to each other using \emph{only} a linear model for the blockade radius; $r_b = \alpha + (C_6/\Omega)^{1/6}$. Observe from \refFigure{fig:paircorrelation} that the Rydberg blockade is clearly at play in each model, since ultimately only a fraction of atoms reaches the excited level. The stable period that is entered after approximately \SI{2.5}{\micro\second} indicates the jamming limit. Note that in the jamming limit, the prediction by each model is near identical. Prior to the jamming limit, the mean number of excited atoms as predicted by models B, C also follows the mean number of excited atoms as predicted by model A roughly; the key difference are the observed quantum oscillations. We can imagine from an experimental viewpoint that if there is a large level of uncertainty in the timestamps of measurements compared to the period of these quantum oscillations, then the measurement curve will resemble the prediction of models B, C. The difference between the predictions of models A and B was 10\% on average in all numerical experiments performed. The difference between the predictions of models A and C was 12\% on average.

\begin{figure}[hbtp]
\centering
\includegraphics[width=0.99\linewidth]{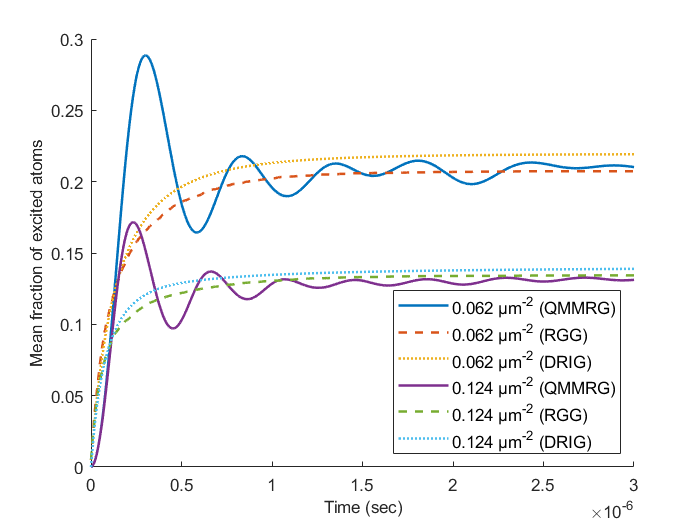}
\caption{The mean fraction of atoms in their excited level for typical values of $\Omega = \SI{5.8}{\mega\hertz} $, $C_6 = \SI{50}{\giga\hertz \micro\meter^6}$ as a function of time, and for two atomic densities, for models A, B and C. These curves were obtained numerically for models A and B through simulation, while for model C we implemented \refTheorem{thm:Fluid_limit_of_Dt}.}
\label{fig:Mean_fraction_of_atoms_as_a_function_of_time_for_models_A_B_C}
\end{figure}

\refFigure{fig:paircorrelation} compares a normalized pair correlation function when measured in the jamming limit of: A. the \gls{QMMRG}, B.\ a \gls{RSA} process on a \gls{RGG}, and D.\ our heuristic of \refSection{sec:Heuristics}. Here, we have temporarily replaced the hard blockade radius in model B specifically with a range-dependent blockade probability that takes the shape $p(r) = 1 / (1 + \exp{(k (r - r_b)))}$. Also, the linear model for the blockade radius has not been used here since we are interested in spatial correlations. Observe that the normalized pair correlation functions are quite similar: the random graph approaches are able to predict the normalized pair-correlation of the quantum mechanical model accurately over the displayed domain. Hence, \refFigure{fig:paircorrelation} also implies that the eventual spatial ordering of excited atoms resulting from either model A or B are similar. The displayed normalized pair-correlation function's shape here is also roughly the same as in e.g.\ \cite[Fig.~2]{robicheaux2005many}; as it should be.

\begin{figure}[hbtp]
\centering
\includegraphics[width=0.99\linewidth]{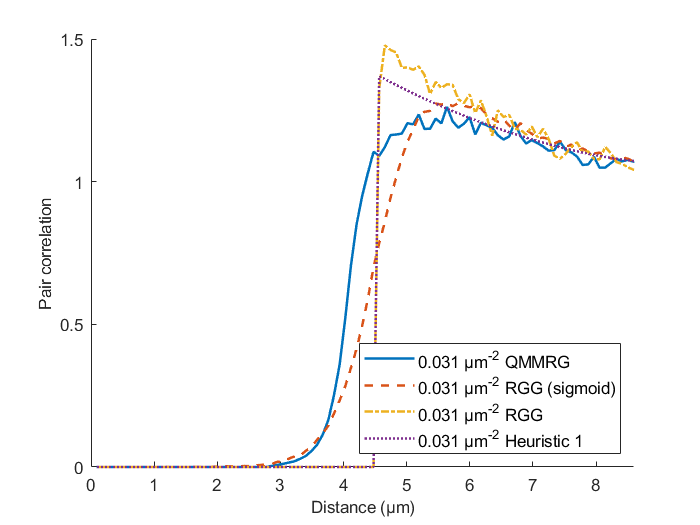}
\caption{A normalized pair correlation function in the jamming limit for the \gls{QMMRG} as compared to the \gls{RGG} and Heuristic 1. Here, the density $\rho = \SI{0.031}{\micro\meter^{-2}}$.}
\label{fig:paircorrelation}
\end{figure}

One remaining, open question raised by \refFigure{fig:paircorrelation} is whether these findings generalize to higher atomic densities. Our goal here was namely to simulate the quantum behavior as accurately as possible, and we forbade ourselves from applying advanced techniques to tackle the numerical complexities inherent in solving the Schr\"{o}dinger equation. This meant that the number of simulated atoms was limited by our computer memory. Another open question is how these results compare to experimental measurements. A quick comparison of our results to e.g.\ the measured pair-correlation function in \cite[Fig.~3]{schauss2012observation} is unfortunately inconclusive. Firstly, the formula for the displayed pair-correlation function there in \cite[(2)]{schauss2012observation} differs from \eqref{eq:paircorrelation} which we display in \refFigure{fig:paircorrelation}. In particular, their denominator is radius dependent, which ties in with this next point: their experimental setup differs from our model setups. The atoms were fixed in position using a square optical lattice instead of being distributed uniformly at random in two-dimensional space, and this affects the statistics. 

\section{Conclusion and Future work}
\label{sec:Conclusion}

We now revisit the question of the relative roles of geometrical restrictions induced by blockade effects and quantum mechanical effects, in the statistics of Rydberg gases. Our numerical comparison of the mean number of excitations over time shows that \gls{RSA} processes on random graphs accurately predict the mean behavior of a quantum mechanical system over time but not the detailed oscillations induced by Schr\"{o}dinger's equation. We conclude that even without dissipation, on long time scales the statistics are affected most by the geometrical restrictions induced by blockade effects, while on short time scales the statistics are affected most by quantum effects. Our heuristic for describing a normalized pair correlation function shows that classical models can even predict the spatial statistics of Rydberg atoms to a considerable degree. We emphasize that no experimental averaging is taking place here: geometrical restrictions are the primary cause. The better mathematical tractability of classical models is an added boon.

If we take a wider viewpoint, we believe that incorporating geometrical restrictions into classical stochastic processes on simpler random graph models may allow for quantitative descriptions of also other statistical phenomena in strongly interacting, ultracold Rydberg gases. For example, \emph{Rydberg aggregates} are networks of Rydberg atoms $\itr{ \vect{X} }{1}, \ldots, \itr{ \vect{X} }{n}$ that define a region $\itr{S}{n} \subseteq \realNumbers^3$ that can nucleate, grow, and eventually reach a long-lived metastable arrangement when the atoms are off-resonantly excited \cite{lesanovsky_out-of-equilibrium_2014}. The off-resonant excitation specifically causes that new Rydberg atoms shift the energy levels of atoms at a specific distance into resonance, and as a consequence after the nucleation of one Rydberg atom, regular arrangements of Rydberg atoms start to grow. We can model the growth mathematically by allowing only particles on the perimeter to become Rydberg atoms: $\itr{ \vect{X} }{n+1} \sim \mathrm{Unif}( \partial \itr{S}{n} )$. The stochastic recursions for this growth process will then be reminiscent of those in \refSection{sec:Results} and \refAppendixSection{sec:Proof_of_the_RIGplusOne_fluid_limit} but more complex due to $\itr{S}{n}$'s shape and the correlations between $\itr{ \vect{X} }{1}, \itr{ \vect{X} }{2}$, etc. We anticipate that the mathematical framework of exploration processes on random graphs may suitably be adapted to quantitatively describe the growth of Rydberg aggregates.

\section*{Acknowledgments}

The authors would like to thank Souvik Dhara, Servaas Kokkelmans, and Debankur Mukherjee for valuable discussion on this topic.

\bibliographystyle{ieeetr}
\bibliography{main}

\onecolumngrid
\appendix

\section{Proof of \refTheorem{thm:Jamming_limit_converges_in_probability}}
\label{sec:Proof_that_the_jamming_limit_converges_in_probability}

The proof of \refTheorem{thm:Jamming_limit_converges_in_probability} follows from two observations. First, the jamming limit $T^*$ equals the total time needed by an exploration process on a \gls{RIG} of the random adsorption type \emph{plus} an exploration process on a graph with isolated vertices. Second, establishing a fluid limit for the former exploration process can be done with a proof almost identical to the proof of \refTheorem{thm:Fluid_limit_of_a_RIG_plus_isolated_vertices_exploration}, which is given in \refAppendixSection{sec:Proof_of_the_RIGplusOne_fluid_limit}. The only difference is that the argumentation of \refAppendixSection{sec:Appendix__Uniform_convergence_in_probability} simplifies: the resulting fluid limit is namely Lipschitz continuous. Having established a fluid limit result for the exploration process on a \gls{RIG} of the random adsorption type---for any $0 \leq S \leq 1, \eps > 0$,
\begin{equation}
\probabilityBig{ \sup_{s \in [0,S]} \Bigl| \frac{ \tilde{U}^\mathrm{RIG}( \lfloor s n^\mathrm{RIG} \rfloor ) }{ n^\mathrm{RIG} } - \tilde{u}(s)  \Bigr| > \eps }
\to 0
\label{eqn:Fluid_limit_of_a_RIG_only}
\end{equation}
as $n \to \infty$ say---one then proceeds as follows. Fix $\eps > 0, S \in [0,1]$ and suppose that at $t = \lfloor S n^{\mathrm{RIG}} \rfloor$, $\tilde{U}( \lfloor Sn^{\mathrm{RIG}} \rfloor ) \geq n^{\mathrm{RIG}} \eps \geq 1$ for otherwise $\tilde{T}^* < \lfloor S n^{\mathrm{RIG}} \rfloor$. Regardless, this implies that at $t = \lfloor Sn^{\mathrm{RIG}} + 1 \rfloor$, the hitting time $\tilde{T}^*$ satisfies the bound
\begin{equation}
( \lfloor Sn^{\mathrm{RIG}} \rfloor + 1 ) \indicator{ n^{\mathrm{RIG}} \eps \geq 1 }
\leq \tilde{T}^* 
\leq \lfloor Sn^{\mathrm{RIG}} \rfloor + \tilde{U}^{\mathrm{RIG}}( \lfloor Sn^{\mathrm{RIG}} \rfloor ) \indicator{ n^{\mathrm{RIG}} \eps\geq 1 }
\quad
\textnormal{w.p.\ one}.
\end{equation}
This is because either the exploration process finishes immediately in the next step (represented by the left-hand side), or finishes in at most $\tilde{U}( \lfloor Sn^{\mathrm{RIG}} \rfloor )$ more iteration steps (represented by the right-hand side). Note now that 
\begin{equation}
\frac{\lfloor Sn^{\mathrm{RIG}} \rfloor}{n^{\mathrm{RIG}}} \indicator{ n^{\mathrm{RIG}} \eps \geq 1 } - \tilde{u}^\gets(0) 
= \frac{\lfloor Sn^{\mathrm{RIG}} \rfloor}{n^{\mathrm{RIG}}} - \tilde{u}^\gets(0) - \frac{\lfloor Sn^{\mathrm{RIG}} \rfloor}{n^{\mathrm{RIG}}} \indicator{ n^{\mathrm{RIG}} \eps < 1 },
\end{equation}
so that the scaled hitting time $\tilde{T}_{n^{\mathrm{RIG}}}^* = \tilde{T}^* / n^{\mathrm{RIG}}$ in particular satisfies the bound
\begin{align}
&
\frac{\lfloor Sn^{\mathrm{RIG}} \rfloor}{n^{\mathrm{RIG}}} - \tilde{u}^\gets(0) - \frac{\lfloor Sn^{\mathrm{RIG}} \rfloor}{n^{\mathrm{RIG}}} \indicator{ n^{\mathrm{RIG}} \eps < 1 } + \frac{1}{n^{\mathrm{RIG}}} \indicator{ n^{\mathrm{RIG}} \eps \geq 1 }
\nonumber \\ &
\leq \tilde{T}_{n^{\mathrm{RIG}}}^* - \tilde{u}^\gets(0) 
\leq \frac{\lfloor Sn^{\mathrm{RIG}} \rfloor}{n^{\mathrm{RIG}}} - \tilde{u}^\gets(0) + \tilde{u}_{n^{\mathrm{RIG}}}(S) \indicator{ n^{\mathrm{RIG}} \eps \geq 1 },
\end{align}
w.p.\ one, and consequently
\begin{equation}
\frac{Sn^{\mathrm{RIG}}-1}{n^{\mathrm{RIG}}} - \tilde{u}^\gets(0) - \indicator{ n^{\mathrm{RIG}} \eps < 1 }
\leq \tilde{T}_{n^{\mathrm{RIG}}}^* - \tilde{u}^\gets(0) 
\leq \frac{Sn^{\mathrm{RIG}}}{n^{\mathrm{RIG}}} - \tilde{u}^\gets(0) + \tilde{u}_{n^{\mathrm{RIG}}}(S)
\quad
\textnormal{w.p.\ one}.
\label{eqn:Intermediate_bound_for_hitting_time_RIG_only}
\end{equation}

We can now prove immediately that for any $\gamma > 0$,
\begin{equation}
\lim_{n^{\mathrm{RIG}} \to \infty} \probability{ | \tilde{T}_{n^{\mathrm{RIG}}}^* - \tilde{u}^\gets(0) | > \gamma }
= 0.
\label{eqn:Almost_final_hitting_time_result_for_a_RIG_only}
\end{equation}
To see this, let $\gamma > 0$. Bound the probability in \eqref{eqn:Almost_final_hitting_time_result_for_a_RIG_only} using \eqref{eqn:Intermediate_bound_for_hitting_time_RIG_only}, specify e.g.\ $\eps < \gamma/2$ and $0 < S = \tilde{u}^\gets(\gamma/2)$, and finally take the limit (while utilizing \eqref{eqn:Fluid_limit_of_a_RIG_only}):
\begin{align}
\probability{ | T_{n^{\mathrm{RIG}}}^* - \tilde{u}^\gets(0) | > \gamma }
&
\eqcom{\ref{eqn:Intermediate_bound_for_hitting_time_RIG_only}}\leq \probabilityBig{ \frac{1}{n^{\mathrm{RIG}}} + \indicator{ n^{\mathrm{RIG}} \eps < 1 } + \tilde{u}_{n^{\mathrm{RIG}}}(S) > \gamma }
\nonumber \\ &
= \probabilityBig{ \frac{1}{n^{\mathrm{RIG}}} + \indicator{ n^{\mathrm{RIG}} \eps < 1 } + \tilde{u}(S) + \tilde{u}_{n^{\mathrm{RIG}}}(S) - \tilde{u}(S) > \gamma }
\nonumber \\ &
\leq \probabilityBig{ \frac{1}{n^{\mathrm{RIG}}} + \indicator{ n^{\mathrm{RIG}} \eps < 1 } + \frac{\gamma}{2} + \eps > \gamma } + \probabilityBig{ \sup_{s \in [0,S]} \Bigl| \frac{ \tilde{U}( \lfloor s n^{\mathrm{RIG}} \rfloor ) }{ n^{\mathrm{RIG}} } - \tilde{u}(s)  \Bigr| > \eps }
\eqcom{\ref{eqn:Fluid_limit_of_a_RIG_only}}\to 0
\end{align}
as $n^{\mathrm{RIG}} \to \infty$. This completes the proof.

\section{Proof of \refTheorem{thm:Fluid_limit_of_a_RIG_plus_isolated_vertices_exploration}}
\label{sec:Proof_of_the_RIGplusOne_fluid_limit}

\refAppendixSection{sec:Proof_of_the_RIGplusOne_fluid_limit} proves a fluid limit for the number of unaffected vertices in a mixture of a \gls{RIG} and a graph consisting solely of isolated vertices. The idea here is to track both the number of unaffected vertices, as well as the number of unaffected attributes in the \gls{RIG} and the number of unaffected vertices in an isolated graph. Their appropriately scaled counterparts can be decomposed in a drift part, which converges to a deterministic function; and a martingale part which vanishes as $n$ tends to infinity when scaled by $n$. 

This method finds its roots in \cite{sanders2015sub,dhara2016generalized,bermolen2017jamming,bermolen2017scaling,dhara2018corrected}, but the approach differs here in that we study an exploration process on \emph{a subgraph of a mixture} of two (random) graphs. This new approach requires us to determine a plethora of fluid limits, one for each of the relevant subgraphs of the mixture of (random) graphs, see \refAppendixSection{sec:Appendix__Uniform_convergence_in_probability}; and then combine these fluid limits using a nonlinear time transformation that depends on an inverse exploration process, see \refAppendixSection{sec:Proof_of_time_rescaling_for_DRIG}. One further caveat is that because we mix two random graphs, a global Lipschitz' continuity property of the fluid limit is lost. This therefore prevents application of the canonical approach of proving a fluid limit by utilizing global Lipschitz' continuity prior to an application of Gr\"{o}nwell's inequality. We circumvent this issue by conditioning only on sample paths that satisfy a local Lipschitz' continuity, and subsequently showing that these sample paths occur with probability one in a large graph limit; see \refAppendixSection{sec:Appendix__Uniform_convergence_in_probability}.

\subsection*{Preliminaries}

Recall that we are given parameters $n^{\mathrm{DRIG}} \in \naturalNumbersPlus$, $\beta, \gamma, c > 0$ according to which we conduct the graph exploration algorithm.

\paragraph{Definitions.}
Let $G^{\mathrm{RIG}} = G_0^{\mathrm{RIG}} \cup G_{1+}^{\mathrm{RIG}}$ be a \gls{RIG} with $n^{\mathrm{RIG}} = \lfloor (1 - e^{-c}) n^{\mathrm{DRIG}} / (1 - \xi_0) \rfloor$ vertices for $\xi_0 = e^{-\beta \gamma (1 - e^{-\gamma})}$ and parameters $m = \lfloor \beta n^{\mathrm{RIG}} \rfloor$ and $p = \gamma / n^{\mathrm{RIG}}$. Let $G_0^{\mathrm{DRIG}}$ be an isolated graph with $\lfloor \e{-c} n^{\mathrm{DRIG}} \rfloor$ vertices. Define $n = n^{\mathrm{RIG}} + \lfloor \e{-c} n^{\mathrm{DRIG}} \rfloor$ for convenience. For any vertex $v \in \mathcal{V}(G^{\mathrm{RIG}}) \cup \mathcal{V}(G_0^{\mathrm{DRIG}})$, let $N_v$ be the \emph{set of neighbors of vertex $v$}, i.e.,
\begin{equation}
\hat{N}_v 
= 
\begin{cases}
\{ w \in \mathcal{V}(G^{\mathrm{RIG}}) \mid \hat{A}_v \cap \hat{A}_w \neq \emptyset, v \neq w \} &\text{if} \quad v \in \mathcal{V}(G^{\mathrm{RIG}}), \\
\emptyset &\text{if} \quad v \in \mathcal{V}(G_0^{\mathrm{DRIG}}).
\end{cases}
\end{equation}
Simultaneously with the \gls{RSA} process, the \gls{DRIG} and following objects will be iteratively constructed:
\begin{itemize}[noitemsep,wide,labelindent=0pt]
\item[--] Let $\hat{X}_0^{\mathrm{RIG}}(t) \subseteq \mathcal{V}(G_0^{\mathrm{RIG}})$ be the \emph{set of isolated, excited vertices at iteration $t$ from $G_0^{\mathrm{RIG}}$}. 
\item[--] Let $\hat{X}_{1+}^{\mathrm{RIG}}(t) \subseteq \mathcal{V}(G_{1+}^{\mathrm{RIG}})$ be the \emph{set of nonisolated, excited vertices at iteration $t$ from $G_{1+}^{\mathrm{RIG}}$}.
\item[--] Let $\hat{X}_0^{\mathrm{DRIG}}(t) \subseteq \mathcal{V}(G_0^{\mathrm{DRIG}})$ be the \emph{set of excited vertices at iteration $t$ from $G_0^{\mathrm{DRIG}}$}. 
\item[--] Let $\hat{U}^{\mathrm{RIG}}(t) \subseteq \mathcal{V}(G^{\mathrm{RIG}})$ be the \emph{set of all unaffected vertices at iteration $t$ from $G^{\mathrm{RIG}}$}. A vertex is \emph{unaffected} if it is neither blocked nor excited.
\item[--] Let $\hat{U}_0^{\mathrm{DRIG}}(t) \subseteq \mathcal{V}(G_0^{\mathrm{DRIG}})$ be the \emph{set of unaffected vertices at iteration $t$ from $G_0^{\mathrm{DRIG}}$}. 
\item[--] Let $\hat{W}^{\mathrm{RIG}}(t) \subseteq \mathcal{A}$ be the \emph{set of all unaffected attributes at iteration $t$ from $G^{\mathrm{RIG}}$}. An attribute is \emph{unaffected} if no vertex connected to it has been excited.
\end{itemize}

\paragraph{Graph exploration algorithm.} Initially, at iteration $t = 0$, set
\begin{equation}
\hat{X}_0^{\mathrm{RIG}}(0) = \emptyset, 
\enskip
\hat{X}_{1+}^{\mathrm{RIG}}(t)(0) = \emptyset,
\enskip
\hat{X}_0^{\mathrm{DRIG}}(0) = \emptyset, 
\enskip
\hat{U}^{\mathrm{RIG}}(0) = \mathcal{V}(G^{\mathrm{RIG}}), 
\enskip
\hat{U}_0^{\mathrm{DRIG}}(0) = \mathcal{V}(G_0^{\mathrm{DRIG}}), 
\enskip
\hat{W}^{\mathrm{RIG}}(0) = \mathcal{A}.
\end{equation}
At each iteration $t \in \{ 1, 2, \ldots, n \}$, a vertex $V(t) \in \hat{U}^{\mathrm{RIG}}(t-1) \cup \hat{U}_0^{\mathrm{DRIG}}(t-1)$ is selected uniformly at random from $\hat{U}^{\mathrm{RIG}}(t-1) \cup \hat{U}_0^{\mathrm{DRIG}}(t-1)$ as long as this set is nonempty. This vertex $V(t)$ is then labeled excited. Thus if $\hat{U}^{\mathrm{RIG}}(t-1) \cup \hat{U}_0^{\mathrm{DRIG}}(t-1)$ is nonempty, then the sets update:
\begin{align}
\hat{X}_0^{\mathrm{RIG}}(t) &= \begin{cases}
    \hat{X}_0^{\mathrm{RIG}}(t-1) \cup \{v\} & \textnormal{if } v \in \mathcal{V}(G^{\mathrm{RIG}}), \hat{N}_v = \emptyset, \\
    \hat{X}_0^{\mathrm{RIG}}(t-1) & \text{otherwise},
\end{cases} \\
\hat{X}_{1+}^{\mathrm{RIG}}(t) &= \begin{cases}
    \hat{X}_{1+}^{\mathrm{RIG}}(t-1) \cup \{v\} & \textnormal{if } v \in \mathcal{V}(G^{\mathrm{RIG}}), \hat{N}_v \neq \emptyset, \\
    \hat{X}_{1+}^{\mathrm{RIG}}(t-1) & \text{otherwise},
\end{cases} \\
\hat{X}_0^{\mathrm{DRIG}}(t) &= \begin{cases}
    \hat{X}_0^{\mathrm{DRIG}}(t-1) \cup \{v\} & \text{if } v \in \mathcal{V}(G_0^{\mathrm{DRIG}}), \\
    \hat{X}_0^{\mathrm{DRIG}}(t-1) & \text{otherwise},
\end{cases} \\
\hat{U}^{\mathrm{RIG}}(t) &= \hat{U}^{\mathrm{RIG}}(t-1) \setminus (\{v\} \cup \hat{N}_v), \\
\hat{U}_0^{\mathrm{DRIG}}(t) &= \hat{U}_0^{\mathrm{DRIG}}(t-1) \setminus \{v\}, \\
\hat{W}^{\mathrm{RIG}}(t) &= \hat{W}^{\mathrm{RIG}}(t-1) \setminus \hat{A}_v.
\label{eqn:Set_evolution}
\end{align}
If otherwise $\hat{U}^{\mathrm{RIG}}(t-1) \cup \hat{U}_0^{\mathrm{DRIG}}(t-1)$ is empty, then the jamming limit has been reached and we set $\hat{X}_0^{\mathrm{RIG}}(t) = \hat{X}_0^{\mathrm{RIG}}(t-1)$, $\hat{X}_{1+}^{\mathrm{RIG}}(t) = \hat{X}_0^{\mathrm{RIG}}(t-1)$, \emph{et cetera}, $\hat{W}^{\mathrm{RIG}}(t+1) = \hat{W}^{\mathrm{RIG}}(t)$.

\subsection{The vector of set sizes constitutes a Markov chain}

In this paper, we only need to investigate the stochastic evolution of the size of each set. We therefore define for $t \in \{ 0, 1, \ldots, n \}$,
\begin{equation}
X_0^{\mathrm{RIG}}(t)
= \lvert \hat{X}_0^{\mathrm{RIG}}(t) \rvert,
\enskip
X_{1+}^{\mathrm{RIG}}(t)
= \lvert \hat{X}_{1+}^{\mathrm{RIG}}(t) \rvert,
\enskip
\textnormal{\emph{et cetera}},
\enskip
W^{\mathrm{RIG}}(t) = \lvert \hat{W}^{\mathrm{RIG}}(t) \rvert,
\end{equation}
and for $t \in \{ 1, \ldots, n \}$,
\begin{equation}
\label{eqn:Definitions_of_dX_dU_dW_etc}
X_0^{\mathrm{RIG}}(t) = X_0^{\mathrm{RIG}}(t-1) + \d{X_0^{\mathrm{RIG}}(t)}, 
\enskip
\textnormal{\emph{et cetera}},
\enskip
W^{\mathrm{RIG}}(t) = W^{\mathrm{RIG}}(t-1) + \d{W^{\mathrm{RIG}}(t)}.
\end{equation}

Now: for $t \in \{ 0, 1, \ldots, n \}$, consider the random vector 
\begin{equation}
E(t) 
= \bigl( X_0^{\mathrm{RIG}}(t), X_{1+}^{\mathrm{RIG}}(t), X_0^{\mathrm{DRIG}}(t), U^{\mathrm{RIG}}(t), U_0^{\mathrm{DRIG}}(t), W^{\mathrm{RIG}}(t) \bigr),
\end{equation}
for which we prove the following.

\begin{lemma}[Markov chain and its drift]
The process $\process{ E(t) }{ t \geq 0 }$ is a Markov chain. Furthermore for $t \in \{ 1, \ldots, n \}$: if $U^{\mathrm{RIG}}(t-1) + U_0^{\mathrm{DRIG}}(t-1) \geq 1$, then
\begin{align}
\expectation{ E(t) | & E(t-1) } - E(t-1)
= 
\nonumber \\ &
\Bigl(
\phantom{-} \frac{U^{\mathrm{RIG}}(t-1)}{U^{\mathrm{RIG}}(t-1) + U_0^{\mathrm{DRIG}}(t-1)} \bigl( 1 - p + p (1 - p)^{n^{\mathrm{RIG}} - X_0^{\mathrm{RIG}}(t-1) - X_{1+}^{\mathrm{RIG}}(t-1)} \bigr)^{W^{\mathrm{RIG}}(t-1)}, 
\label{eqn:Drift_of_dXgi}
\\ &
\phantom{-} \frac{U^{\mathrm{RIG}}(t-1)}{U^{\mathrm{RIG}}(t-1) + U_0^{\mathrm{DRIG}}(t-1)} \bigl( 1 - (1 - p + p (1 - p)^{n^{\mathrm{RIG}} - X_0^{\mathrm{RIG}}(t-1) - X_{1+}^{\mathrm{RIG}}(t-1)} )^{W^{\mathrm{RIG}}(t-1)} \bigr),
\label{eqn:Drift_of_dXgc}
\\ &
\phantom{-} \frac{U_0^{\mathrm{DRIG}}(t-1)}{U^{\mathrm{RIG}}(t-1) + U_0^{\mathrm{DRIG}}(t-1)},
\label{eqn:Drift_of_dXi}
\\ &
- \frac{U^{\mathrm{RIG}}(t-1)}{U^{\mathrm{RIG}}(t-1) + U_0^{\mathrm{DRIG}}(t-1)} (1 + (U^{\mathrm{RIG}}(t-1) - 1)(1 - (1 - p^2)^{W^{\mathrm{RIG}}(t-1)})),
\label{eqn:Drift_of_dUg}
\\ &
- \frac{U_0^{\mathrm{DRIG}}(t-1)}{U^{\mathrm{RIG}}(t-1) + U_0^{\mathrm{DRIG}}(t-1)},
\label{eqn:Drift_of_dUi}
\\ &
- \frac{U^{\mathrm{RIG}}(t-1)}{U_0^{\mathrm{DRIG}}(t-1) + U^{\mathrm{RIG}}(t-1)} p W^{\mathrm{RIG}}(t-1)
\label{eqn:Drift_of_dW}
\Bigr);
\end{align}
otherwise, if $U^{\mathrm{RIG}}(t-1) + U_0^{\mathrm{DRIG}}(t-1) = 0$, then $\expectation{ E(t) | E(t-1) } = E(t-1)$ because \emph{a fortiori} $E(t) = E(t-1)$.
\end{lemma}

\begin{proof}
To prove that $\process{E(t)}{t \geq 0}$ is a Markov chain, we need to show that for all trajectories $e(1), \ldots, e(t+1)$,
\begin{equation}
\probability{ E(t+1) = e(t+1) | E(1) = e(1), E(2) = e(2), \ldots, E(t) = e(t) } 
= \probability{ E(t+1) = e(t+1) | E(t) = e(t) }.
\label{eqn:Markov_chain_property}
\end{equation}
We will split the task of showing \eqref{eqn:Markov_chain_property} into examinations of the individual probability density functions describing the components of the vector $E(t)$. Let $\indicator{ \mathcal{E} }$ be the random variable that takes the value one if the event $\mathcal{E}$ is true, and zero otherwise. Substitute \eqref{eqn:Set_evolution} into \eqref{eqn:Definitions_of_dX_dU_dW_etc}, and note that: if $U^{\mathrm{RIG}}(t) + U_0^{\mathrm{DRIG}}(t) > 0$, then
\begin{align}
\d{X_0^{\mathrm{RIG}}(t)} &= \indicator{ V(t) \in \mathcal{V}(G^{\mathrm{RIG}}) } \indicator{ N_{V(t)} = 0 }, \\
\d{X_{1+}^{\mathrm{RIG}}(t)} &= \indicator{ V(t) \in \mathcal{V}(G^{\mathrm{RIG}}) } ( 1 - \indicator{ N_{V(t)} = 0 } ), \\
\d{X_0^{\mathrm{DRIG}}(t)} &= 1 - \indicator{ V(t) \in \mathcal{V}(G^{\mathrm{RIG}}) }, \\
\d{U^{\mathrm{RIG}}(t)} &= - \indicator{ V(t) \in \mathcal{V}(G^{\mathrm{RIG}}) } ( 1 + N_{V(t)} ), \\
\d{U_0^{\mathrm{DRIG}}(t)} &= - ( 1 - \indicator{ V(t) \in \mathcal{V}(G^{\mathrm{RIG}}) } ), \\
\d{W^{\mathrm{RIG}}(t)} &= - A_{V(t)} \indicator{ V(t) \in \mathcal{V}(G^{\mathrm{RIG}}) };
\label{eqn:Intermediate_expressions_for_dX_dU_dW_etc}
\end{align}
otherwise $\d{X_0^{\mathrm{RIG}}(t)} = \cdots = \d{W^{\mathrm{RIG}}(t)} = 0$. We now proceed component-wise.

\noindent
\emph{Proof that $\d{X_0^{\mathrm{RIG}}(t)} | E(1), \ldots, E(t-1) \eqcom{d}= \d{X_0^{\mathrm{RIG}}(t)} | E(t-1)$:}
First note that by construction,
\begin{equation}
\probability{ V(t) \in \mathcal{V}(G^{\mathrm{RIG}}) | E(1), \ldots, E(t) }
= \frac{U^{\mathrm{RIG}}(t-1)}{U^{\mathrm{RIG}}(t-1)+U_0^{\mathrm{DRIG}}(t-1)} \indicator{ U^{\mathrm{RIG}}(t-1) + U_0^{\mathrm{DRIG}}(t-1) \geq 1 }
\label{eqn:Uniform_vertex_selection}
\end{equation}
because $V(t)$ is selected uniformly at random from $\hat{U}_0^{\mathrm{DRIG}}(t-1) \cup \hat{U}^{\mathrm{RIG}}(t-1)$, independently. Note in particular that \eqref{eqn:Uniform_vertex_selection} also implies that $\probability{ V(t) \in \mathcal{V}(G^{\mathrm{RIG}}) \mid E(1), \ldots, E(t) } = \probability{ V(t) \in \mathcal{V}(G^{\mathrm{RIG}}) \mid E(t) }$---a fact that we will use later---but that this is not sufficient for proving the current claim.

To prove the current claim, note that $\d{X_0^{\mathrm{RIG}}(t)} \in \{ 0, 1 \}$. Therefore, using \eqref{eqn:Intermediate_expressions_for_dX_dU_dW_etc} and (i) the law of total probability, we find that
\begin{align}
&
\probability{ \d{X_0^{\mathrm{RIG}}(t)} = 1 | E(1), \ldots, E(t-1) }
\nonumber \\ &
\eqcom{\ref{eqn:Intermediate_expressions_for_dX_dU_dW_etc}}=
\probability{ V(t) \in \mathcal{V}(G^{\mathrm{RIG}}), N_{V(t)} = 0 | E(1), \ldots, E(t-1) }
\nonumber \\ &
\eqcom{i}= \probability{ V(t) \in \mathcal{V}(G^{\mathrm{RIG}}), N_{V(t)} = 0 | E(1), \ldots, E(t-1), V(t) \in \mathcal{V}(G^{\mathrm{RIG}}) } \probability{ V(t) \in \mathcal{V}(G^{\mathrm{RIG}}) | E(1), \ldots, E(t-1) }
\nonumber \\ &
\phantom{=}+ \probability{ V(t) \in \mathcal{V}(G^{\mathrm{RIG}}), N_{V(t)} = 0 | E(1), \ldots, E(t-1), V(t) \not\in \mathcal{V}(G^{\mathrm{RIG}}) } \probability{ V(t) \not\in \mathcal{V}(G^{\mathrm{RIG}}) | E(1), \ldots, E(t-1) }
\nonumber \\ &
= \probability{ V(t) \in \mathcal{V}(G^{\mathrm{RIG}}) | E(1), \ldots, E(t-1) } \probability{ N_{V(t)} = 0 | E(1), \ldots, E(t-1), V(t) \in \mathcal{V}(G^{\mathrm{RIG}}) }.
\label{eqn:Decomposition_of_dXgis_PDF_using_the_law_of_total_probability}
\end{align}
Now by construction of the graph exploration algorithm, we have that
\begin{align}
&
\probability{ N_{V(t)} = 0 | E(1), \ldots, E(t-1), V(t) \in \mathcal{V}(G^{\mathrm{RIG}}) }
\nonumber \\ &
= \expectationBig{ (1-p)^{A_{V(t)} (n - X_0^{\mathrm{RIG}}(t-1) - X_{1+}^{\mathrm{RIG}}(t-1))} \Big| E(1), \ldots, E(t-1), V(t) \in \mathcal{V}(G^{\mathrm{RIG}}) } 
\nonumber \\ &
= \bigl( 1 - p + p (1 - p)^{n^{\mathrm{RIG}} - X_0^{\mathrm{RIG}}(t-1) - X_{1+}^{\mathrm{RIG}}(t-1)} \bigr)^{W^{\mathrm{RIG}}(t-1)}.
\label{eqn:Probability_of_zero_neighbors_in_a_RIG}
\end{align}
Substituting \eqref{eqn:Uniform_vertex_selection} and \eqref{eqn:Probability_of_zero_neighbors_in_a_RIG} into \eqref{eqn:Decomposition_of_dXgis_PDF_using_the_law_of_total_probability} implies that
\begin{equation}
\probability{ \d{X_0^{\mathrm{RIG}}(t)} = 1 | E(1), \ldots, E(t-1) }
= \probability{ \d{X_0^{\mathrm{RIG}}(t)} = 1 | E(t-1) },
\end{equation}
with the right-hand side here being equal to \eqref{eqn:Drift_of_dXgi} because $\d{X_0^{\mathrm{RIG}}(t)} \in \{ 0, 1 \}$. Hence, this substitution also proves \eqref{eqn:Drift_of_dXgi}.

\noindent
\emph{Proofs that $\d{X_0^{\mathrm{DRIG}}(t)} | E(1), \ldots, E(t-1) \eqcom{d}= \d{X_{0}^{\mathrm{DRIG}}(t)} | E(t-1)$ and $\d{U_0^{\mathrm{DRIG}}(t)} | E(1), \ldots, E(t-1) \eqcom{d}= \d{U_0^{\mathrm{DRIG}}(t)} | E(t-1)$:} The first claim follows immediately from \eqref{eqn:Uniform_vertex_selection}; see the discussion below it. Because $\d{X_0^{\mathrm{DRIG}}(t)} \in \{0,1\}$, \eqref{eqn:Uniform_vertex_selection} also implies \eqref{eqn:Drift_of_dXi}. Combined with linearity, the exact same argumentation also proves the second claim as well as \eqref{eqn:Drift_of_dUi}.

\noindent
\emph{Proof that $\d{X_{1+}^{\mathrm{RIG}}(t)} | E(1), \ldots, E(t-1) \eqcom{d}= \d{X_{1+}^{\mathrm{RIG}}(t)} | E(t-1)$:} Combine the two previous arguments; linearity proves the current claim. On account of $\d{X_{1+}^{\mathrm{RIG}}(t)} \in \{0,1\}$, this approach also proves \eqref{eqn:Drift_of_dXgc}.

\noindent
\emph{Proof that $\d{U^{\mathrm{RIG}}(t)} | E(1), \ldots, E(t-1) \eqcom{d}= \d{U^{\mathrm{RIG}}(t)} | E(t-1)$:} Let $k \in \naturalNumbersPlus$. By (i) the law of total probability,
\begin{align}
&
\probability{ \d{U^{\mathrm{RIG}}(t)} = - k | E(1), \ldots, E(t-1) }
\nonumber \\ &
\eqcom{\ref{eqn:Intermediate_expressions_for_dX_dU_dW_etc}}= \probability{ V(t) \in \mathcal{V}(G^{\mathrm{RIG}}), N_{V(t)} = k - 1 | E(1), \ldots, E(t-1) } 
\nonumber \\ &
\eqcom{i}= \probability{ V(t) \in \mathcal{V}(G^{\mathrm{RIG}}), N_{V(t)} = k - 1 | E(1), \ldots, E(t-1), V(t) \in \mathcal{V}(G^{\mathrm{RIG}}) } \probability{ V(t) \in \mathcal{V}(G^{\mathrm{RIG}}) | E(1), \ldots, E(t-1) } 
\nonumber \\ &
\phantom{=}+ \probability{ V(t) \in \mathcal{V}(G^{\mathrm{RIG}}), N_{V(t)} = k - 1 | E(1), \ldots, E(t-1), V(t) \not\in \mathcal{V}(G^{\mathrm{RIG}}) } \probability{  V(t) \not\in \mathcal{V}(G^{\mathrm{RIG}}) | E(1), \ldots, E(t-1) }
\nonumber \\ &
= \probability{ N_{V(t)} = k - 1 | E(1), \ldots, E(t-1), V(t) \in \mathcal{V}(G^{\mathrm{RIG}}) } \probability{ V(t) \in \mathcal{V}(G^{\mathrm{RIG}}) | E(1), \ldots, E(t-1) }.
\label{eqn:Prob_dU_k__Law_of_total_probability}
\end{align}
Note now that by construction, for $v \in \hat{U}^{\mathrm{RIG}}(t-1)$,
\begin{align}
&
\probability{ N_{V(t)} = k - 1 | E(1), \ldots, E(t-1), V(t) \in \mathcal{V}(G^{\mathrm{RIG}}), V(t) = v }
\nonumber \\ &
= \sum_{a=0}^{W^{\mathrm{RIG}}(t-1)} \probability{ N_{V(t)} = k - 1 | E(1), \ldots, E(t-1), V(t) = v, |A_v| = a } \probability{ |A_v| = a | E(1), \ldots, E(t-1), v \in \mathcal{V}(G^{\mathrm{RIG}}) }
\nonumber \\ &
= \sum_{a=0}^{W^{\mathrm{RIG}}(t-1)} \binom{U^{\mathrm{RIG}}(t-1)-1}{k-1} ( 1 - (1-p)^a )^{k-1} (1-p)^{a ( U^{\mathrm{RIG}}(t-1) - k ) } \binom{ W^{\mathrm{RIG}}(t-1) }{a} p^a (1-p)^{ W^{\mathrm{RIG}}(t-1) - a }.
\label{eqn:Prob_NV_km1__Case_I}
\end{align}
Substituting \eqref{eqn:Uniform_vertex_selection} and \eqref{eqn:Prob_NV_km1__Case_I} into \eqref{eqn:Prob_dU_k__Law_of_total_probability}, we find that for all $k \in \naturalNumbersPlus$,
\begin{equation}
\probability{ \d{U^{\mathrm{RIG}}(t)} = k | E(1), \ldots, E(t-1) }
= \probability{ \d{U^{\mathrm{RIG}}(t)} = k | E(t-1) }.
\end{equation}
Here, the right-hand side is implicitly given by the expression that results after the substitution.

We now proceed and calculate the drift. By (i) symmetry, and (ii) after substitution of \eqref{eqn:Prob_NV_km1__Case_I} and algebraic manipulation, we find that
\begin{align}
&
\expectation{ \d{U^{\mathrm{RIG}}(t)} | E(t-1)} 
\nonumber \\ &
= - \frac{U^{\mathrm{RIG}}(t-1)}{U^{\mathrm{RIG}}(t-1) + U_0^{\mathrm{DRIG}}(t-1)} 
\nonumber \\ &
\phantom{=}
\times \sum_{k=1}^{U^{\mathrm{RIG}}(t-1)} k  \sum_{v \in \hat{U}^{\mathrm{RIG}}(t-1)} \probability{ N_{V(t)} = k - 1 | E(1), \ldots, E(t-1), V(t) \in \mathcal{V}(G^{\mathrm{RIG}}), V(t) = v }
\nonumber \\ &
\eqcom{i,ii}= - \frac{U^{\mathrm{RIG}}(t-1)}{U^{\mathrm{RIG}}(t-1) + U_0^{\mathrm{DRIG}}(t-1)} \bigl( 1 - \bigl( U^{\mathrm{RIG}}(t-1) - 1 \bigr) \bigl( 1 - (1-p^2)^{W^{\mathrm{RIG}}(t-1)} \bigr) \bigr).
\end{align}
This proves \eqref{eqn:Drift_of_dUg}.

\noindent
\emph{Proof that $\d{W^{\mathrm{RIG}}(t)} | E(1), \ldots, E(t-1) \eqcom{d}= \d{W^{\mathrm{RIG}}(t)} | E(t-1)$:} Let $k \in \naturalNumbersZero$. By (i) the law of total probability, and (ii) symmetry of the vertices $v \in \hat{U}^{\mathrm{RIG}}(t-1)$,
\begin{align}
&
\probability{ \d{W^{\mathrm{RIG}}(t)} = - k | E(1), \ldots, E(t-1) }
\nonumber \\ &
\eqcom{\ref{eqn:Intermediate_expressions_for_dX_dU_dW_etc}}= \probability{ A_{V(t)} = k, V(t) \in \mathcal{V}(G^{\mathrm{RIG}}) | E(1), \ldots, E(t-1) }
\nonumber \\ &
\eqcom{i}= \sum_{ v \in \hat{U}^{\mathrm{RIG}}(t-1) \cup \hat{U}_0^{\mathrm{DRIG}}(t-1) } \probability{ A_{V(t)} = k, V(t) \in \mathcal{V}(G^{\mathrm{RIG}}) | E(1), \ldots, E(t-1), V(t) = v } \probability{ V(t) = v | E(1), \ldots, E(t-1) }
\nonumber \\ &
\eqcom{ii}= \frac{ U^{\mathrm{RIG}}(t-1) }{ U^{\mathrm{RIG}}(t-1) + U_0^{\mathrm{DRIG}}(t-1) } \binom{W^{\mathrm{RIG}}(t-1)}{k} p^k (1-p)^{W^{\mathrm{RIG}}(t-1)-k}.
\label{eqn:Prob_dW_k_and_Vt_in_G}
\end{align}
Consequently, for all $k \in \naturalNumbersZero$,
\begin{equation}
\probability{ \d{W^{\mathrm{RIG}}(t)} = -k, | E(1), \ldots, E(t-1) }
= \probability{ \d{W^{\mathrm{RIG}}(t)} = -k | E(t-1) }.
\end{equation}
This proves the current claim.

Inspecting \eqref{eqn:Prob_dW_k_and_Vt_in_G}, we can conclude that
\begin{align}
&
- \d{W^{\mathrm{RIG}}(t)} | E(t-1), V(t) \in \mathcal{V}(G^{\mathrm{RIG}}) 
\eqcom{d}= \mathrm{Binomial}( W^{\mathrm{RIG}}(t-1), p ),
\quad
\textnormal{and}
\nonumber \\ &
- \d{W^{\mathrm{RIG}}(t)} | E(t-1), V(t) \not\in \mathcal{V}(G^{\mathrm{RIG}}) 
\eqcom{d}= 0 \textnormal{ w.p.\ one}.
\end{align}
These two observations imply \eqref{eqn:Drift_of_dW}, which concludes the proof.
\end{proof}

\subsection{Martingale decompositions and integral representations of the elements of \texorpdfstring{$X_0^{\mathrm{RIG}}(t), X_{1+}^{\mathrm{RIG}}(t), \ldots, W^{\mathrm{RIG}}(t)$}{processes}}

For $s \in [0,1]$, define scaled variants of the set sizes by
\begin{equation}
x_{0,n}^{\mathrm{RIG}}(s) = \frac{X_0^{\mathrm{RIG}}(\floor{s n})}{n},
\enskip
x_{1+,n}^{\mathrm{RIG}}(s) = \frac{X_{1+}^{\mathrm{RIG}}(\floor{s n})}{n},
\enskip
x_{0,n}^{\mathrm{DRIG}}(s) = \frac{X_0^{\mathrm{DRIG}}(\floor{s n})}{n}, 
\enskip
\textnormal{\emph{et cetera}},
\enskip
w_n^{\mathrm{RIG}}(s) = \frac{W^{\mathrm{RIG}}(\floor{s n})}{n}.
\label{eqn:Definition_of_the_scaled_processes}
\end{equation}

\noindent
\emph{Integral representation of $X_0^{\mathrm{RIG}}(t)$.}
Using Doob--Meyer's decomposition theorem \cite{karatzas_2012} on e.g.\ $X_0^{\mathrm{RIG}}(t)$, we find the martingale decomposition
\begin{align}
X_0^{\mathrm{RIG}}(t) 
&
= \sum\limits_{i=1}^t \d{X_0^{\mathrm{RIG}}(t)}
\label{eqn:Xgi_martingale_decomposition}
\\ &
= M_n^{X_0^{\mathrm{RIG}}}(t) + \sum\limits_{i=1}^t \frac{ U^{\mathrm{RIG}}(i-1) \indicator{U^{\mathrm{RIG}}(i-1) + U_0^{\mathrm{DRIG}}(i-1) \geq 1} }{U^{\mathrm{RIG}}(i-1) + U_0^{\mathrm{DRIG}}(i-1)}
\nonumber \\ &
\phantom{= M_n^{X_0^{\mathrm{RIG}}}(t) + \sum\limits_{i=1}^t} \times \bigl( 1 - p + p (1 - p)^{n^{\mathrm{RIG}} - X_0^{\mathrm{RIG}}(i-1) - X_{1+}^{\mathrm{RIG}}(i-1)} \bigr)^{W^{\mathrm{RIG}}(i-1)},
\nonumber
\end{align}
with $M_n^{X_0^{\mathrm{RIG}}}(t)$ a square-integrable martingale. To see why, note that 
\begin{equation}
M_n^{X_0^{\mathrm{RIG}}}(t) 
= \sum_{i=1}^t \bigl( X_0^{\mathrm{RIG}}(i) - X_0^{\mathrm{RIG}}(i-1) - \expectation{ \d{X_0^{\mathrm{RIG}}(i)} | Z_{i-1} } \bigr) 
= \sum_{i=1}^t \bigl( \d{X_0^{\mathrm{RIG}}(i)} - \expectation{ \d{X_0^{\mathrm{RIG}}(i)} | Z_{i-1} } \bigr).
\end{equation}
by telescoping. Clearly $\expectation{ ( M_n^{X_0^{\mathrm{RIG}}}(t) )^2 } \leq 4 t^2 \leq 4 n^2 < \infty$ by construction of the exploration algorithm, and hypothesis.
Substituting \eqref{eqn:Xgi_martingale_decomposition} into \eqref{eqn:Definition_of_the_scaled_processes}, we find that the scaled process $x_{0,n}^{\mathrm{RIG}}(s)$ can be written to leading order as an integral:
\begin{align}
&
x_{0,n}^{\mathrm{RIG}}(s) 
\nonumber \\ &
= \frac{M_n^{X_0^{\mathrm{RIG}}}(\floor{s n})}{n} + \frac{1}{n} \sum\limits_{i=1}^{\floor{s n}} \frac{ \frac{U^{\mathrm{RIG}}(i-1)}{n} \cdot \indicator{U^{\mathrm{RIG}}(i-1) + U_0^{\mathrm{DRIG}}(i-1) \geq 1} }{  \frac{U^{\mathrm{RIG}}(i-1)}{n} + \frac{U_0^{\mathrm{DRIG}}(t-1)}{n} } 
\nonumber \\ &
\phantom{= \frac{M_n^{X_0^{\mathrm{RIG}}}(\floor{s n})}{n} + \frac{1}{n} \sum\limits_{i=1}^{\floor{s n}}} \times \bigl( 1 - p + p (1 - p)^{ n ( \frac{\sigma}{\sigma + e^{-c}} - \frac{X_0^{\mathrm{RIG}}(i-1)}{n} - \frac{X_{1+}^{\mathrm{RIG}}(i-1)}{n} ) } \bigr)^{W^{\mathrm{RIG}}(i-1)} 
\nonumber \\ &
= \frac{M_n^{X_0^{\mathrm{RIG}}}(\floor{s n})}{n} + \int\limits_0^s \frac{ u_n^{\mathrm{RIG}}(x) \indicator{u_n^{\mathrm{RIG}}(x) + u_{0,n}^{\mathrm{DRIG}}(x) \geq 1/n} }{u_n^{\mathrm{RIG}}(x) + u_{0,n}^{\mathrm{DRIG}}(x)} 
\nonumber \\ &
\phantom{= \frac{M_n^{X_0^{\mathrm{RIG}}}(\floor{s n})}{n} + \int\limits_0^s}
\times \exp{ \Bigl( -\tilde{\gamma} w_n^{\mathrm{RIG}}(x) \bigl( 1 - \e{ -\tilde{\gamma} ( \frac{\sigma}{\sigma + \e{-c}} - x_{0,n}^{\mathrm{RIG}}(t) - x_{1+,n}^{\mathrm{RIG}}(t) ) } \bigr) \Bigr) } \d{x} + o_\mathbb{P}(1),
\end{align}
where $\tilde{\gamma} = (\sigma + \e{-c}) \gamma / \sigma$.

\noindent
\emph{Integral representations of $X_{1+}^{\mathrm{RIG}}(t)$ and $X_0^{\mathrm{DRIG}}(t)$.}
Similarly, we can directly obtain integral representations of $X_{1+}^{\mathrm{RIG}}(t)$ and $X_0^{\mathrm{DRIG}}(t)$'s scaled processes up to leading order. That is
\begin{align}
x_{1+,n}^{\mathrm{RIG}}(s) 
&
= \frac{M_n^{X_{1+}^{\mathrm{RIG}}}(\floor{s n})}{n} + \int\limits_0^s \frac{ u_n^{\mathrm{RIG}}(x) \indicator{u_n^{\mathrm{RIG}}(x) + u_{0,n}^{\mathrm{DRIG}}(x) \geq 1/n} }{u_n^{\mathrm{RIG}}(x) + u_{0,n}^{\mathrm{DRIG}}(x)} 
\nonumber \\ &
\phantom{= \frac{M_n^{X_{1+}^{\mathrm{RIG}}}(\floor{s n})}{n} + \int\limits_0^s} \times \Bigl( 1 - \exp{ \bigl( -\tilde{\gamma} w_n^{\mathrm{RIG}}(x) ( 1 - \e{ -\tilde{\gamma} ( \frac{\sigma}{\sigma + \e{-c}} - x_{0,n}^{\mathrm{RIG}}(t) - x_{1+,n}^{\mathrm{RIG}}(t) ) } ) \bigr) } \Bigr) \d{x} + o_\mathbb{P}(1), 
\nonumber \\
x_{0,n}^{\mathrm{DRIG}}(s) 
&
= \frac{M_n^{X_0^{\mathrm{DRIG}}}(\floor{s n})}{n} + \int\limits_0^s \frac{ u_{0,n}^{\mathrm{DRIG}}(x) \indicator{u_n^{\mathrm{RIG}}(x) + u_{0,n}^{\mathrm{DRIG}}(x) \geq 1/n} }{u_n^{\mathrm{RIG}}(x) + u_{0,n}^{\mathrm{DRIG}}(x)} \d{x},
\end{align}
where $M_n^{X_{1+}^{\mathrm{RIG}}}(t)$ and $M_n^{X_0^{\mathrm{DRIG}}}(t)$ are again square-integrable martingales. 

\noindent
\emph{Integral representation of $U^{\mathrm{RIG}}(t)$.}
Executing the same methodology for $U^{\mathrm{RIG}}(t)$ yields
\begin{align}
U^{\mathrm{RIG}}(t) 
&
= n^{\mathrm{RIG}} + \sum\limits_{i=1}^t \d{U^{\mathrm{RIG}}(t)} 
\nonumber \\ &
= n^{\mathrm{RIG}} + M_n^{U^{\mathrm{RIG}}}(t) - \sum\limits_{i=1}^t  \frac{ \frac{U^{\mathrm{RIG}}(i-1)}{n} \cdot \indicator{U^{\mathrm{RIG}}(i-1) + U_0^{\mathrm{DRIG}}(i-1) \geq 1} }{ \frac{U^{\mathrm{RIG}}(i-1)}{n} + \frac{U_0^{\mathrm{DRIG}}(t-1)}{n} } 
\nonumber \\ &
\phantom{= n^{\mathrm{RIG}} + M_n^{U^{\mathrm{RIG}}}(t) - \sum\limits_{i=1}^t} \times \bigl( 1 + (U^{\mathrm{RIG}}(i-1) - 1)(1 - (1 - p^2 )^{W^{\mathrm{RIG}}(i-1)}) \bigr),
\label{eqn:Summand_for_Ugt}
\end{align}
where $M_n^{U^{\mathrm{RIG}}}(t)$ is once more a square-integrable martingale. Now, expand the multiplicative term in the summand, and note that $W^{\mathrm{RIG}}(t) \leq \lfloor \beta n \rfloor $ w.p.\ one, to conclude that
\begin{equation}
1 - (1 - p^2)^{W^{\mathrm{RIG}}(i-1)} = 1 - (1 - p^2 W^{\mathrm{RIG}}(i-1) + O_\mathbb{P}(p^4 W^{\mathrm{RIG}}(i-1)^2)) = \frac{\tilde{\gamma}^2 W^{\mathrm{RIG}}(i-1)}{n^2} + O_\mathbb{P}(n^{-2}).
\end{equation}
Consequently, $U^{\mathrm{RIG}}$'s scaled process $u_n^{\mathrm{RIG}}(s)$ has, to leading order, the integral representation
\begin{align}
u_n^{\mathrm{RIG}}(s) 
&
= \frac{\sigma}{\sigma + e^{-c}} + \frac{M_n^{U^{\mathrm{RIG}}}(\floor{s n})}{n} - \frac{1}{n} \sum\limits_{i=1}^{\floor{s n}} \frac{ \frac{U^{\mathrm{RIG}}(i-1)}{n} \cdot \indicator{U^{\mathrm{RIG}}(i-1) + U_0^{\mathrm{DRIG}}(i-1) \geq 1} }{  \frac{U^{\mathrm{RIG}}(i-1)}{n} + \frac{U_0^{\mathrm{DRIG}}(t-1)}{n} } 
\nonumber \\ &
\phantom{= \frac{\sigma}{\sigma + e^{-c}} + \frac{M_n^{U^{\mathrm{RIG}}}(\floor{s n})}{n} - \frac{1}{n} \sum\limits_{i=1}^{\floor{s n}}}
\times \Bigl( 1 + \frac{U_{i-1}}{n} \frac{\tilde{\gamma}^2 W^{\mathrm{RIG}}_{i-1}}{n} \Bigr) + O_\mathbb{P}(n^{-1}) 
\nonumber \\ &
= \frac{\sigma}{\sigma + e^{-c}} + \frac{M_n^{U^{\mathrm{RIG}}}(\floor{s n})}{n} - \int\limits_0^s \frac{ u_n^{\mathrm{RIG}}(x) \indicator{u_n^{\mathrm{RIG}}(x) + u_{0,n}^{\mathrm{DRIG}}(x) \geq 1/n} }{u_n^{\mathrm{RIG}}(x) + u_{0,n}^{\mathrm{DRIG}}(x)} (1 + \tilde{\gamma}^2 u_n(x) w_n^{\mathrm{RIG}}(x)) \d{x} + O_\mathbb{P}(n^{-1}).
\label{eqn:Prelimit_integral_representation_of_unRIG}
\end{align}

\noindent
\emph{Integral representation of $U_0^{\mathrm{DRIG}}(t)$.}
It will now be little surprise that the scaled process $u_{0,n}^{\mathrm{DRIG}}(s)$ is given by
\begin{equation}
u_{0,n}^{\mathrm{DRIG}}(s) 
= \frac{e^{-c}}{\sigma + e^{-c}} = + \frac{M_n^{U_0^{\mathrm{DRIG}}}(\floor{s n})}{n} - \int\limits_0^s \frac{ u_{0,n}^{\mathrm{DRIG}}(x) \indicator{u_n^{\mathrm{RIG}}(x) + u_{0,n}^{\mathrm{DRIG}}(x) \geq 1/n} }{u_n^{\mathrm{RIG}}(x) + u_{0,n}^{\mathrm{DRIG}}(x)} \d{x},
\label{eqn:Prelimit_integral_representation_of_u0nDRIG}
\end{equation}
where $M_n^{U^{\mathrm{RIG}}}(t)$ is another square-integrable martingale.

\noindent
\emph{Integral representation of $W^{\mathrm{RIG}}(t)$.}
Finally, the Doob--Meyer decomposition of $W^{\mathrm{RIG}}(t)$ yields the martingale decomposition
\begin{equation}
W^{\mathrm{RIG}}(t) 
= m + \sum\limits_{i=1}^t \d{W^{\mathrm{RIG}}(t)} 
= m + M_n^w(t) - \sum\limits_{i=1}^t \frac{ U^{\mathrm{RIG}}(i-1) \indicator{U^{\mathrm{RIG}}(i-1) + U_0^{\mathrm{DRIG}}(i-1) \geq 1} }{U^{\mathrm{RIG}}(i-1) + U_0^{\mathrm{DRIG}}(i-1)} p W^{\mathrm{RIG}}(i-1),
\end{equation}
where $M_n^w(t)$ is (one more time) a square-integrable martingale. Its scaled partner variable is therefore given by
\begin{align}
w_n^{\mathrm{RIG}}(s) 
&
= \frac{m}{n} + \frac{M_n^w(\floor{s n})}{n} - \frac{1}{n} \sum\limits_{i=1}^{\floor{s n}} p W^{\mathrm{RIG}}(t-1) 
\nonumber \\ &
= \frac{\beta \sigma}{\sigma + e^{-c}} + \frac{M_n^w(\floor{s n})}{n} - \int\limits_0^s \frac{ u_{0,n}^{\mathrm{DRIG}}(x) \indicator{u_n^{\mathrm{RIG}}(x) + u_{0,n}^{\mathrm{DRIG}}(x) \geq 1/n} }{u_n^{\mathrm{RIG}}(x) + u_{0,n}^{\mathrm{DRIG}}(x)} \tilde{\gamma} w_n^{\mathrm{RIG}}(x) \d{x}.
\label{eqn:Prelimit_integral_representation_of_wnRIG}
\end{align}

\subsection{Quadratic variation and convergence of the martingales}

To prove the almost sure convergence of the scaled processes $x_{0,n}^{\mathrm{RIG}}, x_{1+,n}^{\mathrm{RIG}}, \ldots$ to $x_0^{\mathrm{RIG}}, x_{1+}^{\mathrm{RIG}}, \ldots$ as claimed in \refTheorem{thm:Fluid_limit_of_a_RIG_plus_isolated_vertices_exploration}, we will investigate the \emph{quadratic variation} of each of the martingales $M_n^{x_0^{\mathrm{RIG}}}, M_n^{x_{1+}^{\mathrm{RIG}}}, \ldots$. These quadratic variations are given by
\begin{equation}
\langle M_n^{x_0^{\mathrm{RIG}}} \rangle (t) 
= \sum\limits_{i=1}^t \mathbb{V}_{i-1}[ \d{X_0^{\mathrm{RIG}}(i)} ], 
\enskip
\textnormal{\emph{et cetera}},
\enskip
\langle M_n^w \rangle (t) 
= \sum\limits_{i=1}^t \mathbb{V}_{i-1}[ \d{W^{\mathrm{RIG}}(i)} ]
\label{eqn:Definitions_of_the_quadratic_variations}
\end{equation}
where we have introduced the short-hand notations $\mathbb{V}_t[ \cdot ] \triangleq \mathbb{V}[ \cdot | E(t) ]$, $\mathbb{E}_t[ \cdot ] \triangleq \expectation{ \cdot | E(t) }$ for these conditional expectations.

Concretely, we prove the following in this section:

\begin{lemma}[Convergence of martingales]
\label{lem:Convergence_of_martingales}
For any $\alpha > 1/2$ it holds that
\begin{equation}
\frac{1}{n^\alpha} \sup_{s \in [0,1]} \lvert M_n^{x_0^{\mathrm{RIG}}}(\floor{s n}) \rvert 
\xrightarrow{\mathbb{P}} 0, 
\quad
\frac{1}{n^\alpha} \sup_{s \in [0,1]} \lvert M_n^{x_{1+}^{\mathrm{RIG}}}(\floor{s n}) \rvert \xrightarrow{\mathbb{P}} 0, 
\quad
\textnormal{\emph{et cetera}},
\quad
\frac{1}{n^\alpha} \sup_{s \in [0,1]} \lvert M_n^w(\floor{s n}) \rvert \xrightarrow{\mathbb{P}} 0
\end{equation}
as $n \to \infty$.
\end{lemma}

\begin{proof}
We are going to prove that for any $s \in [0, 1]$,
\begin{gather}
\langle M_n^{x_0^{\mathrm{RIG}}} \rangle (\floor{s n}) \leq n, 
\enskip
\langle M_n^{x_{1+}^{\mathrm{RIG}}} \rangle (\floor{s n}) \leq n, 
\enskip
\langle M_n^{x_0^{\mathrm{DRIG}}} \rangle (\floor{s n}) \leq n, 
\nonumber \\
\langle M_n^{u^{\mathrm{RIG}}} \rangle (\floor{s n}) \leq \beta \gamma^2 (1 + \gamma) n, 
\enskip
\langle M_n^{u_0^{\mathrm{DRIG}}} \rangle (\floor{s n}) \leq n, 
\enskip
\langle M_n^w \rangle (\floor{s n}) \leq \jaron{ \beta \gamma ( 1 + \beta \gamma ) } n
\label{eqn:Absolute_bounds_on_the_quadratic_variations}
\end{gather}
with probability one. The result then follows after an application of Doob's inequality \cite{liptser_2012}, because e.g.\
\begin{equation}
\probabilityBig{ \frac{1}{n^\alpha} \sup_{s \in [0,1]} \lvert M_n^{x_0^{\mathrm{RIG}}}(\floor{s n}) \rvert \geq \varepsilon }
\leq \frac{1}{\varepsilon^2 n^{2\alpha}} \expectation{ ( M_n^{gi}(n) )^2 }
= \frac{1}{\varepsilon^2 n^{2\alpha}} \expectation{ \langle M_n^{gi}(n) \rangle }
\end{equation} 
and therefore, under \eqref{eqn:Absolute_bounds_on_the_quadratic_variations}, for all $\varepsilon > 0$
\begin{equation}
\lim_{n \to \infty} \probabilityBig{ \frac{1}{n^\alpha} \sup_{s \in [0,1]} \lvert M_n^{x_0^{\mathrm{RIG}}} \rvert \geq \varepsilon }
\leq \lim_{n \to \infty} \frac{1}{n^{2\alpha-1} \varepsilon^2} 
= 0.
\end{equation}
The remaining claims follow \emph{mutatis mutandis}.

\noindent
\emph{Proof that $\langle M_n^{x_0^{\mathrm{RIG}}}( \lfloor sn \rfloor ) \rangle$, $\langle M_n^{x_{1+}^{\mathrm{RIG}}}( \lfloor sn \rfloor ) \rangle$, $\langle M_n^{x_0^{\mathrm{DRIG}}}( \lfloor sn \rfloor ) \rangle$, $\langle M_n^{u_0^{\mathrm{DRIG}}}( \lfloor sn \rfloor ) \rangle \leq n$ w.p.\ one.}
Note that 
\begin{equation}
\mathbb{V}_{i-1}[ \d{X_0^{\mathrm{RIG}}(i)} ], 
\mathbb{V}_{i-1}[ \d{x_{1+}^{\mathrm{RIG}}(i)} ], 
\mathbb{V}_{i-1}[ \d{X_0^{\mathrm{DRIG}}(i)} ], 
\mathbb{V}_{i-1}[ \d{U_0^{\mathrm{DRIG}}(i)} ] 
\leq 1
\quad
\textnormal{w.p.\ one}
\end{equation}
because $\d{X_0^{\mathrm{RIG}}(t)}, \d{X_{1+}^{\mathrm{RIG}}(t)}, \d{X_0^{\mathrm{DRIG}}(t)}$, $\d{U_0^{\mathrm{DRIG}}(t)} \in \{0, 1\}$. Bound \eqref{eqn:Definitions_of_the_quadratic_variations} directly and the claim follows.

\noindent
\emph{Proof that $\langle M_n^{u^{\mathrm{RIG}}}( \lfloor sn \rfloor ) \rangle \leq \beta \gamma^2 (1 + \gamma) n$ w.p.\ one.}
Consider a summand in the quadratic variation of $M_n^{u^{\mathrm{RIG}}}(t)$ in \eqref{eqn:Definitions_of_the_quadratic_variations}. The law of total variance implies for it that
\begin{equation}
\mathbb{V}_{t-1}[ \d{U^{\mathrm{RIG}}(t)} ] 
= \ecnd{ \mathbb{V}_{t-1}[ \d{U^{\mathrm{RIG}}(t)} \mid A_{V(t)} ] } + \mathbb{V}_{t-1}[ \ecnd{ \d{U^{\mathrm{RIG}}(t)} \mid A_{V(t)} } ].
\label{eqn:dUg__Law_of_total_variance}
\end{equation}
Recall now that by construction of the exploration algorithm, $\d{U^{\mathrm{RIG}}(t)} | E(t-1), A_{V(t)}$ is binomially distributed with parameters $U^{\mathrm{RIG}}(t-1) - 1$ and $(1-p)^{A_{V(t)}}$ if $U^{\mathrm{RIG}}(t-1) + U_0^{\mathrm{DRIG}}(t-1) \geq 1$. Recollect in addition that for $n \in \naturalNumbersZero, p \in [0,1]$, $\expectation{ \mathrm{Binomial}(n,p)} = np$ and $\mathbb{V}[ \mathrm{Binomial}(n,p) ] = np(1-p)$. If $U^{\mathrm{RIG}}(t-1) + U_0^{\mathrm{DRIG}}(t-1) \geq 1$, then the first term in \eqref{eqn:dUg__Law_of_total_variance} expands therefore as
\begin{align}
\ecnd{ \mathbb{V}_{t-1}[ \d{U^{\mathrm{RIG}}(t)} \mid A_{V(t)} ]} 
&
= ( U^{\mathrm{RIG}}(t-1) - 1 ) \ecnd{ (1-p)^{A_{V(t)}} (1 - (1-p)^{A_{V(t)}}) } 
\nonumber \\ &
= ( U^{\mathrm{RIG}}(t-1) - 1 ) \bigl( \ecnd{ (1-p)^{A_{V(t)}} } - \ecnd{ (1-p)^{ 2 A_{V(t)} }} \bigr) 
\nonumber \\ &
= ( U^{\mathrm{RIG}}(t-1) - 1 ) \bigl( (1-p + p(1-p))^{ W^{\mathrm{RIG}}(t-1) } - ( 1-p + p(1-p)^2 )^{W^{\mathrm{RIG}}(t-1)} \bigr) 
\nonumber \\ &
= ( U^{\mathrm{RIG}}(t-1) - 1 ) \bigl( (1-p^2)^{ W^{\mathrm{RIG}}(t-1) } - (1-2p^2+p^3)^{ W^{\mathrm{RIG}}(t-1) } \bigr),
\label{eqn:dUg__Law_of_total_variance__Term_I}
\end{align}
and the second term in \eqref{eqn:dUg__Law_of_total_variance} as
\begin{align}
\mathbb{V}_{t-1}[ \ecnd{ \d{U^{\mathrm{RIG}}(t)} \mid A_{V(t)} } ] 
&
= ( U^{\mathrm{RIG}}(t-1) - 1 )^2 \mathbb{V}_{t-1}[ (1-p)^{ A_{V(t)} } ] 
\nonumber \\ &
= ( U^{\mathrm{RIG}}(t-1) - 1 )^2 \bigl( \ecnd{(1-p)^{2A_{V(t)}}} - \ecnd{(1-p)^{A_{V(t)}}}^2 \bigr) 
\nonumber \\ & 
= ( U^{\mathrm{RIG}}(t-1) - 1 )^2 \bigl( (1-2p^2+p^3)^{W^{\mathrm{RIG}}(t-1)} - (1-p^2)^{2W^{\mathrm{RIG}}(t-1)} \bigr) 
\nonumber \\ & 
= ( U^{\mathrm{RIG}}(t-1) - 1 )^2 \bigl( (1-2p^2+p^3)^{W^{\mathrm{RIG}}(t-1)} - (1-2p^2+p^4)^{W^{\mathrm{RIG}}(t-1)} \bigr).
\label{eqn:dUg__Law_of_total_variance__Term_II}
\end{align}
Note now the algebraic (in)equalities that (i) for all $a, b \in \realNumbers$ and $m \in \naturalNumbersPlus$, $a^m - b^m = (a-b) \sum_{i=1}^m a^{m-i} b^{i-1}$, $U^{\mathrm{RIG}}(t-1) \leq n$ w.p.\ one, (ii) $0 \leq 1-2p^2+p^3 \leq 1 - p^2 \leq 1$, and finally (iii) $W^{\mathrm{RIG}}(t-1) \leq \lfloor \beta n^{\mathrm{RIG}} \rfloor$ w.p.\ one. We therefore have
\begin{align}
\ecnd{ \mathbb{V}_{t-1}( \d{U^{\mathrm{RIG}}(t)} \mid A_{V(t)}) } 
&
\eqcom{i}\leq n^{\mathrm{RIG}} \indicator{W^{\mathrm{RIG}}(t)>0} p^2 (1-p) \sum_{i=1}^{W^{\mathrm{RIG}}(t)} (1-p^2)^{W^{\mathrm{RIG}}(t)-i} (1-2p^2+p^3)^{i-1}
\nonumber \\ &
\eqcom{ii}\leq n^{\mathrm{RIG}} p^2(1-p) W^{\mathrm{RIG}}(t) 
\eqcom{iii}\leq \frac{\gamma^2 \lfloor \beta n^{\mathrm{RIG}} \rfloor }{n^{\mathrm{RIG}}}
\leq \gamma^2 \beta
\quad
\textnormal{w.p.\ one}.
\end{align}
Similarly
$
\mathbb{V}_{t-1}[ \ecnd{ \d{U^{\mathrm{RIG}}(t)} \mid A_{V(t)} } ]
\leq \gamma^3 \beta
$
with probability one. This proves the claim.

\noindent
\emph{Proof that $\langle M_n^{w}( \lfloor sn \rfloor ) \rangle \leq \beta \gamma (1 + \beta \gamma) n$ w.p.\ one.} Finally, if $U^{\mathrm{RIG}}(t-1) + U_0^{\mathrm{DRIG}}(t-1) \geq 1$, the summand in the quadratic variation of $M_n^w(t)$ in \eqref{eqn:Definitions_of_the_quadratic_variations} satisfies
\begin{align}
&
\mathbb{V}_{t-1}[ \d{W^{\mathrm{RIG}}(t)} ] 
\nonumber \\ &
\eqcom{\ref{eqn:Intermediate_expressions_for_dX_dU_dW_etc}}= \ecnd{ \mathrm{Var}_{t-1}[ \d{W^{\mathrm{RIG}}(t)} | V(t) \in \mathcal{V}(G^{\mathrm{RIG}}) ] } 
+ \mathrm{Var}_{t-1}[ \ecnd{ \d{W^{\mathrm{RIG}}(t)} | V(t) \in \mathcal{V}(G^{\mathrm{RIG}}) } ]
\nonumber \\ &
= \ecnd{ \mathrm{Var}_{t-1}[ - A_{V(t)} \indicator{ V(t) \in \mathcal{V}(G^{\mathrm{RIG}}) } | V(t) \in \mathcal{V}(G^{\mathrm{RIG}}) ] } 
\nonumber \\ &
\phantom{=}+ \mathrm{Var}_{t-1}[ \ecnd{ - A_{V(t)} \indicator{ V(t) \in \mathcal{V}(G^{\mathrm{RIG}}) } | V(t) \in \mathcal{V}(G^{\mathrm{RIG}}) } ]
\nonumber \\ &
= \ecnd{ \indicator{ V(t) \in \mathcal{V}(G^{\mathrm{RIG}}) } W^{\mathrm{RIG}}(t-1) p(1-p) } + \mathrm{Var}_{t-1}[ - \indicator{ V(t) \in \mathcal{V}(G^{\mathrm{RIG}}) } W^{\mathrm{RIG}}(t-1) p ]
\nonumber \\ &
= \frac{U^{\mathrm{RIG}}(t-1)}{U_0^{\mathrm{DRIG}}(t-1) + U^{\mathrm{RIG}}(t-1)} W^{\mathrm{RIG}}(t-1) p(1-p) + ( W^{\mathrm{RIG}}(t-1) )^2 p^2 \mathrm{Var}_{t-1}[ \d{X_0^{\mathrm{DRIG}}}(t) - 1 ]
\nonumber \\ &
\leq \frac{\gamma}{n^{\mathrm{RIG}}} \Bigl( 1 - \frac{\gamma}{n^{\mathrm{RIG}}} \Bigr) \lfloor \beta n^{\mathrm{RIG}} \rfloor + \Bigl( \frac{\gamma}{n^{\mathrm{RIG}}} \Bigr)^2 \lfloor \beta n^{\mathrm{RIG}} \rfloor^2 
\leq \frac{ \gamma \lfloor \beta n^{\mathrm{RIG}} \rfloor }{n^{\mathrm{RIG}}} + \frac{ \gamma^2 \lfloor \beta n^{\mathrm{RIG}} \rfloor^2 }{(n^{\mathrm{RIG}})^2}
\leq \beta \gamma ( 1 + \beta \gamma )
\quad
\textnormal{w.p.\ one}.
\end{align}
This completes the proof.
\end{proof}

\subsection{Uniform convergence in probability of the scaled processes}
\label{sec:Appendix__Uniform_convergence_in_probability}

Recall our definitions for $x_0^{\mathrm{RIG}}, x_{1+}^{\mathrm{RIG}}, x_0^{\mathrm{DRIG}}, u^{\mathrm{RIG}}, u_0^{\mathrm{DRIG}}, w$ in \eqref{eq:limits_processes}. We are going to prove the following:

\begin{lemma}[Uniform convergence in probability]
\label{lem:Uniform_convergence_in_probability_of_the_scaled_processes}
For any $0 \leq S < \phi$, $\eps > 0$:
\begin{equation}
\probability{ \sup_{s \in [0,S]} | x_{0,n}^{\mathrm{RIG}}(s) - x_0^{\mathrm{RIG}}(s) | \geq \eps }
\to 0,
\enskip
\textnormal{\emph{et cetera}},
\enskip
\probability{ \sup_{s \in [0,S]} | w_n^{\mathrm{RIG}}(s) - w(s) | \geq \eps }
\to 0
\end{equation}
as $n \to \infty$.
\end{lemma}

\noindent
\emph{Proof that for $0 \leq S < \phi$, $u_n^{\mathrm{RIG}} \to u^{\mathrm{RIG}}$, $u_{0,n}^{\mathrm{DRIG}} \to u_0^{\mathrm{DRIG}}$ and $w_n^{\mathrm{RIG}} \to w$ uniformly on $[0,S]$ in probability.}
Start by noting that the vector-valued function $f(x_1,x_2,x_3) = \bigl( x_1 (1 + \tilde{\gamma}^2 x_1 x_3), x_2, \tilde{\gamma} x_1 x_3 \bigr) / (x_1 + x_2)$ is Lipschitz continuous on $\mathcal{D}_\eps = \{ (x_1,x_2) \in [0,1] | x_1 + x_2 \geq \eps \} \times [ 0, {\beta \sigma} / {\sigma + e^{-c}} ]$ for all $\eps > 0$. \emph{Viz.}, 
\begin{equation}
\forall \eps > 0 \, \exists L_\eps > 0 
: 
\pnorm{ f(x_1,x_2,x_3) - f(y_1,y_2,y_3) }{1}
\leq L_\eps \pnorm{ (x_1,x_2,x_3) - (y_1,y_2,y_3) }{1}
\, \forall x, y \in \mathcal{D}_\eps.
\end{equation}
Also note that consequently, for the vector-valued function $h_n(x_1,x_2,x_3) = f(x_1,x_2,x_3) \indicator{ x_1 \geq 1/n }$, it holds that
\begin{align}
\forall \eps > 0 \exists L_\eps > 0 
:
\pnorm{ h_n(x) - f(y) }{1}
&
\leq \pnorm{ h_n(x) - f(x) }{1} + \pnorm{ f(x) - f(y) }{1} 
\nonumber \\ &
\leq C_{\beta,\tilde{\gamma}} \indicator{ x_1 + x_2 \in [\eps,1/n) } + L_\eps \pnorm{x-y}
\, \forall x, y \in \mathcal{D}_\eps.
\label{eqn:Lipschitz_continuity_and_one_over_n_difference}
\end{align}
for some absolute constant $C_{\beta,\tilde{\gamma}}$ (i.e., independent of $n$).

Now let $0 \leq S < \phi$, and choose $0 < \delta < (u^{\mathrm{RIG}}+u_0^{\mathrm{DRIG}})(S)$, and $0 < \eps_\delta \leq (u^{\mathrm{RIG}}+u_0^{\mathrm{DRIG}})(S) - \delta$. 
Consider the events
\begin{equation}
\Delta_n(\delta)
= \bigl\{ \sup_{s \in [0,S]} \pnorm{(u_n^{\mathrm{RIG}},u_{0,n}^{\mathrm{DRIG}},w_n^{\mathrm{RIG}})(s) - (u^{\mathrm{RIG}},u_0^{\mathrm{DRIG}},w)(s) }{1} < \delta \bigr\}, 
\quad
\mathcal{E}_n(\eps)
= \bigl\{ \inf_{s \in [0,S]} (u_n^{\mathrm{RIG}}+u_{0,n}^{\mathrm{DRIG}})(s) \geq \eps \bigr\}.
\end{equation}
Critically, note that $\forall \omega \in \Delta_n(\delta) : \forall s \in [0,S], (u_n^{\mathrm{RIG}}+u_{0,n}^{\mathrm{DRIG}})(s) = u_n^{\mathrm{RIG}}(s)-u^{\mathrm{RIG}}(s) + u_{0,n}^{\mathrm{DRIG}}(s)-u_0^{\mathrm{DRIG}}(s) + u^{\mathrm{RIG}}(s)+u_0^{\mathrm{DRIG}}(s) \geq - |u_n^{\mathrm{RIG}}(s)-u^{\mathrm{RIG}}(s)| - |u_{0,n}^{\mathrm{DRIG}}(s)-u_0^{\mathrm{DRIG}}(s)| + u^{\mathrm{RIG}}(s)+u_0^{\mathrm{DRIG}}(s) \geq u^{\mathrm{RIG}}(s)+u_0^{\mathrm{DRIG}}(s) - \delta \geq (u^{\mathrm{RIG}}+u_0^{\mathrm{DRIG}})(S) - \delta \geq \eps_\delta > 0$ since $u^{\mathrm{RIG}}+u_0^{\mathrm{DRIG}}$ is nonincreasing, and by hypothesis. Hence $\Delta_n(\delta) \subseteq \mathcal{E}_n(\eps_\delta)$. 

Because of \eqref{eqn:Prelimit_integral_representation_of_unRIG}, \eqref{eqn:Prelimit_integral_representation_of_u0nDRIG}, \eqref{eqn:Prelimit_integral_representation_of_wnRIG}, and \eqref{eqn:Lipschitz_continuity_and_one_over_n_difference}, we have that for arbitrary $\eps > 0$, it holds on the event $\mathcal{E}_n(\eps)$ that 
\begin{align}
&
\sup_{s \in [0,S]} \pnorm{ (u_n^{\mathrm{RIG}},u_{0,n}^{\mathrm{DRIG}},w_n^{\mathrm{RIG}})(s) - (u^{\mathrm{RIG}},u_0^{\mathrm{DRIG}},w)(s) }{1}
\nonumber \\ &
\eqcom{\ref{eqn:Prelimit_integral_representation_of_unRIG},\ref{eqn:Prelimit_integral_representation_of_u0nDRIG},\ref{eqn:Prelimit_integral_representation_of_wnRIG}}\leq 
\sup_{s \in [0,S]} \frac{\lvert M_n^{u^{\mathrm{RIG}}}(\floor{n s}) \rvert}{n} 
+ \sup_{s \in [0,S]} \frac{\lvert M_n^{u_0^{\mathrm{DRIG}}}(\floor{n s}) \rvert}{n} 
+ \sup_{s \in [0,S]} \frac{\lvert M_n^{w}(\floor{n s}) \rvert}{n} 
\nonumber \\ &
\phantom{\leq}+ \sup_{s \in [0,S]} \int\limits_0^s 
\Bigl\lVert
\Bigl( 
\frac{u_n^{\mathrm{RIG}}(s) ( 1 + \tilde{\gamma}^2 u_n^{\mathrm{RIG}}(s) w_n^{\mathrm{RIG}}(s) ) \indicator{u_n^{\mathrm{RIG}}(s)+u_{0,n}^{\mathrm{DRIG}}(s) \geq 1/n}}{u_n^{\mathrm{RIG}}(s) + u_{0,n}^{\mathrm{DRIG}}(s)} - \frac{ u^{\mathrm{RIG}}(s) (1 + \tilde{\gamma}^2 u^{\mathrm{RIG}}(s) ) }{u^{\mathrm{RIG}}(s) + u_0^{\mathrm{DRIG}}(s)}
,
\nonumber \\ &
\phantom{\leq + \sup_{s \in [0,S]} \int\limits_0^s \Bigl\lVert \Bigl(}
\frac{u_n^{\mathrm{RIG}}(s) \indicator{u_n^{\mathrm{RIG}}(s)+u_{0,n}^{\mathrm{DRIG}}(s) \geq 1/n}}{u_n^{\mathrm{RIG}}(s) + u_{0,n}^{\mathrm{DRIG}}(s)} - \frac{u^{\mathrm{RIG}}(s)}{u^{\mathrm{RIG}}(s) + u_0^{\mathrm{DRIG}}(s)}
,
\nonumber \\ &
\phantom{\leq + \sup_{s \in [0,S]} \int\limits_0^s \Bigl\lVert \Bigl(}
\frac{u_n^{\mathrm{RIG}}(s) \indicator{u_n^{\mathrm{RIG}}(s)+u_{0,n}^{\mathrm{DRIG}}(s) \geq 1/n}}{u_n^{\mathrm{RIG}}(s) + u_{0,n}^{\mathrm{DRIG}}(s)} \tilde{\gamma} w_n^{\mathrm{RIG}}(s) - \frac{u^{\mathrm{RIG}}(s)}{u^{\mathrm{RIG}}(s) + u_0^{\mathrm{DRIG}}(s)} \tilde{\gamma} w(s)
\Bigr) 
\Bigr\rVert_1 \d{x} 
\nonumber \\ &
\leq \textnormal{martingale terms} + \int_0^S \sup_{s \in [0,x]} \pnorm{ h_n( u_n^{\mathrm{RIG}}(s),u_{0,n}^{\mathrm{DRIG}}(s), w_n^{\mathrm{RIG}}(s) ) - f(u^{\mathrm{RIG}}(s),u_0^{\mathrm{DRIG}}(s),w^{\mathrm{RIG}}(s)) }{1} \d{x}
\nonumber \\ &
\eqcom{\ref{eqn:Lipschitz_continuity_and_one_over_n_difference}}\leq \textnormal{martingale terms} + C_{\beta, \tilde{\gamma}} ) \indicator{ (u_n^{\mathrm{RIG}}+u_{0,n}^{\mathrm{DRIG}})(S) \in [\eps,1/n) } 
\nonumber \\ &
\phantom{\eqcom{\ref{eqn:Lipschitz_continuity_and_one_over_n_difference}}\leq}+ L_\eps \int_0^S \sup_{s \in [0,x]} \pnorm{ (u_n^{\mathrm{RIG}},u_{0,n}^{\mathrm{DRIG}},w_n^{\mathrm{RIG}})(s) - (u^{\mathrm{RIG}},u_0^{\mathrm{DRIG}},w^{\mathrm{RIG}})(s) }{1} \d{x},
\label{eqn:ug_ui__Bound_prior_to_Gronwall}
\end{align}
where the last inequality also relies on the fact that $u_n^{\mathrm{RIG}}+u_{0,n}^{\mathrm{DRIG}}$ is nonincreasing and $S \leq 1$. Applying Gr\"{on}wall's inequality to \eqref{eqn:ug_ui__Bound_prior_to_Gronwall} reveals that on the event $\mathcal{E}_n(\eps)$,
\begin{align}
&
\sup_{s \in [0,S]} \pnorm{ (u_n^{\mathrm{RIG}},u_{0,n}^{\mathrm{DRIG}},w_n^{\mathrm{RIG}})(s) - (u^{\mathrm{RIG}},u_0^{\mathrm{DRIG}},w^{\mathrm{RIG}})(s) }{1}
\nonumber \\ &
\leq \bigl( \textnormal{martingale terms} + C_{\beta,\tilde{\gamma}} \indicator{ (u_n^{\mathrm{RIG}}+u_{0,n}^{\mathrm{DRIG}})(S) \in [\eps,1/n) } \bigr) \e{L_\eps}.
\label{eqn:ug_ui__Gronwall_on_event_En}
\end{align}

We now tie the arguments up to \eqref{eqn:ug_ui__Gronwall_on_event_En} together: for any $0 \leq S < \phi$, $0 < \delta < (u^{\mathrm{RIG}}+u_0^{\mathrm{DRIG}})(S)$ and $0 < \eps_\delta \leq \delta$, all independent of $n$,
\begin{align}
&
\probability{ \sup_{s \in [0,S]} \pnorm{ (u_n^{\mathrm{RIG}},u_{0,n}^{\mathrm{DRIG}},w_n^{\mathrm{RIG}})(s) - (u^{\mathrm{RIG}},u_0^{\mathrm{DRIG}},w^{\mathrm{RIG}})(s) }{1} < \delta }
\nonumber \\ &
\eqcom{\ref{eqn:ug_ui__Gronwall_on_event_En}}\geq \probabilityBig{ \textnormal{martingale terms} + C_{\beta,\tilde{\gamma}} \indicator{ (u_n^{\mathrm{RIG}}+u_{0,n}^{\mathrm{DRIG}})(S) \in [\eps_\delta,1/n) } < \delta \e{-L_{\eps_\delta}} }
\quad \because \Delta_n(\delta) \subseteq \mathcal{E}_n(\eps_\delta)
\nonumber \\ &
\to 1 
\end{align}
as $n \to \infty$ by \refLemma{lem:Convergence_of_martingales} \emph{if} $\probability{ \mathcal{E}_n(\eps_\delta) } \to 1$. 

One final step therefore remains, that is, to prove that $\probability{ \mathcal{E}_n(\eps_\delta) } \to 1$. Assume that $n \eps > 1$. Note that
$
T^* > \lfloor S n \rfloor + n \eps
$
implies 
$
U^{\mathrm{RIG}}( \lfloor S n \rfloor ) + U_0^{\mathrm{DRIG}}( \lfloor S n \rfloor ) \geq n \eps.
$
Therefore
\begin{equation}
T_n^* > \frac{ \lfloor S n \rfloor }{n} + \eps 
\Rightarrow \mathcal{E}_n(\eps).
\end{equation}
Hence as long as $S + \eps < \phi$, then
\begin{equation}
\probability{ \mathcal{E}_n(\eps) }
\geq \probabilityBig{ T_n^* > \frac{ \lfloor S n \rfloor }{n} + \eps }
\geq \probability{ T_n^* > S + \eps } \to 1
\end{equation}
by \refTheorem{thm:Jamming_limit_converges_in_probability}. This proves the claim.

\noindent
\emph{Proof that for $0 \leq S < \phi$, $x_{0,n}^{\mathrm{DRIG}} \to x_0^{\mathrm{DRIG}}$ uniformly on $[0,S]$ in probability.}
Because $x_{0,n}^{\mathrm{DRIG}} = 1 - u_{0,n}^{\mathrm{DRIG}}$ and $x_0^{\mathrm{DRIG}} = 1 - u_0^{\mathrm{DRIG}}$, the claim follows from the arguments above.

\noindent
\emph{Proof that for $0 \leq S < \phi$, $x_{0,n}^{\mathrm{RIG}} \to x_0^{\mathrm{RIG}}$ and $x_{1+,n}^{\mathrm{RIG}} \to x_{1+}^{\mathrm{RIG}}$ uniformly on $[0,S]$ in probability.} The proof approach is almost identical to what was done to prove that $u_n^{\mathrm{RIG}} \to u^{\mathrm{RIG}}$ and $u_{0,n}^{\mathrm{DRIG}} \to u_0^{\mathrm{DRIG}}$ uniformly on $[0,S]$ in probability. The only necessary change is to observe first that the vector-valued function 
\begin{equation}
g(x_1,x_2,x_3,x_4,x_5) 
= \Bigl( 
f(x_1,x_2,x_3), 
\frac{x_1}{x_1+x_2} \e{ - \gamma x_3 (1 - \e{-\gamma(1-x_4-x_5)}) },
\frac{x_1}{x_1+x_2} \bigl( 1 - \e{ - \gamma x_3 (1 - \e{-\gamma(1-x_4-x_5)}) 
} \bigr)
\Bigr)
\end{equation}
is Lipschitz continuous on $\mathcal{D}_\eps^g = \{ x_1, x_2 \in [0,1] | x_1 + x_2 \geq \eps \} \times \left[0,\frac{\beta \sigma}{\sigma + e^{-c}}\right] \times [0,1]^2$, say with constant $K_\eps^g$, and continue the arguments from there (\emph{mutatis mutandis}). This concludes the proof. \qed

\section{Proof of \refTheorem{thm:Fluid_limit_of_Dt}}
\label{sec:Proof_of_time_rescaling_for_DRIG}

\refAppendixSection{sec:Proof_of_time_rescaling_for_DRIG} converts the fluid limit results of \refAppendixSection{sec:Proof_of_the_RIGplusOne_fluid_limit} to a fluid limit result for \gls{RSA} on a \gls{DRIG} by applying a nonlinear time transformation. Recall that for $t \in \{ 0, 1, \ldots, n \}$, 
\begin{equation}
X^{\mathrm{DRIG}}(t) = X_{1+}^{\mathrm{RIG}}(t) + X_0^{\mathrm{DRIG}}(t), 
\quad
U^{\mathrm{DRIG}}(t) 
= U_0^{\mathrm{DRIG}}(t) + U^{\mathrm{RIG}}(t) + X_0^{\mathrm{RIG}}(t) - X_0^{\mathrm{RIG}}(n), 
\end{equation}
and 
\begin{equation}
D^{\mathrm{DRIG}}(t) 
= U^{\mathrm{DRIG}}((X^{\mathrm{DRIG}})^\gets(t)).
\end{equation}

First, \refLemma{lem:Uniform_convergence_in_probability_of_the_scaled_processes} implies that for $0 \leq S < \phi$, $X^{\mathrm{DRIG}}(\lfloor s n \rfloor) / n$ converges to $x(s) = (x_{1+}^{\mathrm{RIG}}+x_0^{\mathrm{DRIG}})(s)$ uniformly on $[0,S]$ in probability as $n \to \infty$. Note furthermore that because $x$ is continuous and nondecreasing, its inverse exists; denote it by $x^\gets$. Recall that 
\begin{equation}
(X^{\mathrm{DRIG}})^\gets(u) 
= \inf \{ t \in \{ 0, 1, \ldots, n \} | X^{\mathrm{DRIG}}(t) = u \}
\quad
\textnormal{for}
\quad
u \in \{ 0, 1, \ldots, n^{\mathrm{DRIG}} \}.
\end{equation}
Let $r \in [0,1]$, and observe that
\begin{align}
(X^{\mathrm{DRIG}})^\gets( \lfloor r n^{\mathrm{DRIG}} \rfloor ) 
&
= \inf \{ t \in \{ 0, 1, \ldots, n \} | X^{\mathrm{DRIG}}(t) = \lfloor r n^{\mathrm{DRIG}} \rfloor \}
\nonumber \\ &
= \inf \{ \lfloor f n \rfloor | X^{\mathrm{DRIG}}(\lfloor f n \rfloor) = \lfloor r n^{\mathrm{DRIG}} \rfloor, f \in [0,1] \}
\nonumber \\ &
= \Bigl\lfloor n \inf \Bigl\{ f \in [0,1] \Big| \frac{ X^{\mathrm{DRIG}}(\lfloor f n \rfloor) }{n} = \frac{\lfloor r (n^{\mathrm{DRIG}}/n) n \rfloor}{n} \Bigr\} \Bigr\rfloor.
\end{align}
This implies that for any $0 \leq R \leq 1$, $\eps > 0$,
\begin{equation}
\probabilityBig{ \sup_{r \in [0,R]} \left\lvert \frac{(X^{\mathrm{DRIG}})^\gets( \lfloor rn^{\mathrm{DRIG}} \rfloor )}{ n } - x^\gets( rn^{\mathrm{DRIG}}/n ) \right\rvert \geq \eps }
\to 0
\label{eqn:Uniform_convergence_of_time_scaling}
\end{equation}
as $n^{\mathrm{DRIG}} \to \infty$. Second, note that \refLemma{lem:Uniform_convergence_in_probability_of_the_scaled_processes} implies that for $0 \leq S < \phi$, $U^{\mathrm{DRIG}}(\lfloor sn \rfloor) / n$ converges to $u^{\mathrm{DRIG}}(s) = (u_0^{\mathrm{DRIG}} + u^{\mathrm{RIG}} + x_0^{\mathrm{RIG}})(s) - \xi_0$ uniformly on $[0,S]$ in probability as $n \to \infty$. 

Let $0 \leq R < \eta$ and $r \in [0,R]$. Using the triangle inequality,
\begin{align}
\Bigl| \frac{D( \lfloor r n^{\mathrm{DRIG}} \rfloor )}{ n^{\mathrm{DRIG}} } - d(r) \Bigl|
&
= \Bigl| \frac{ U^{\mathrm{DRIG}}( (X^{\mathrm{DRIG}})^\gets( \lfloor rn^{\mathrm{DRIG}} \rfloor ) ) }{ n^{\mathrm{DRIG}} } - \frac{ u^{\mathrm{DRIG}}( x^\gets( r n^{\mathrm{DRIG}}/n ) }{ n^{\mathrm{DRIG}}/n } \Bigr|
\nonumber \\ &
\leq 
\Bigl| \frac{ U^{\mathrm{DRIG}}( (X^{\mathrm{DRIG}})^\gets( \lfloor rn^{\mathrm{DRIG}} \rfloor ) ) }{ n^{\mathrm{DRIG}} } - \frac{ U^{\mathrm{DRIG}}( n x^\gets(r n^{\mathrm{DRIG}}/n) }{ n^{\mathrm{DRIG}} } ) \Bigr| 
\nonumber \\ &
\phantom{\leq}
+ \Bigl| \frac{ U^{\mathrm{DRIG}}( n x^\gets(r n^{\mathrm{DRIG}}/n) ) }{ n (n^{\mathrm{DRIG}}/n) } - \frac{ u^{\mathrm{DRIG}}( x^\gets(rn^{\mathrm{DRIG}}/n) ) }{ n^{\mathrm{DRIG}}/n } \Bigr|.
\end{align}
By \eqref{eqn:Uniform_convergence_of_time_scaling} and continuity of $U^{\mathrm{DRIG}}(\cdot)$, the first term converges to zero in probability as $n^{\mathrm{DRIG}} \to \infty$; recall that $n^{\mathrm{DRIG}} \leq n$. The second term converges to zero by uniform convergence of $U^{\mathrm{DRIG}}(\lfloor s n \rfloor)/n$ to $u^{\mathrm{DRIG}}(s)$ as $n^{\mathrm{DRIG}} \to \infty$ since then also $n \to \infty$. This last step is allowed, because $x^\gets(r n^{\mathrm{DRIG}}/n)
< \phi$ as long as $r < \eta$.

\section{Simulation of the quantum mechanical system}\label{app:simulation}

Here, we briefly summarize our numerical implementation of the \gls{QMMRG} model. The simulation was created implemented in Matlab in order to efficiently solve the system of differential equations \eqref{eqn:Schrodingers_equation}. 

First of all, we uniformly random distribute the $N$ atoms over a two-dimensional torus. The positions of the atoms are saved in a matrix $r \in \mathbbm{R}^{N \times 2}$. The next step is to compute the state space. A state is represented by a logical vector $s \in \{0, 1\}^N$. The state space is then truncated to a set of reasonably reachable states to tackle the exponential increase of memory consumption in the number of atoms. Our numerical experiments have shown that only a fraction of the states in the state space actually contribute in the solution of the Schrödinger equation as Figure \ref{fig:probstate}(a) exemplifies. Consider for example the simulation with 8 atoms which has $2^8 = 256$ states in total. The state with the highest contribution has a probability of 10\% of occurring in the jamming limit. Moreover, with probability at least 99\% the system is in one of the first 44 states, which is only 17\% of the total number of states. To capture the majority of the behavior of the quantum system it therefore seems enough to only use a fraction of the state space.

\begin{figure}
\centering 
\subfloat[]{
  \includegraphics[width=0.46\textwidth]{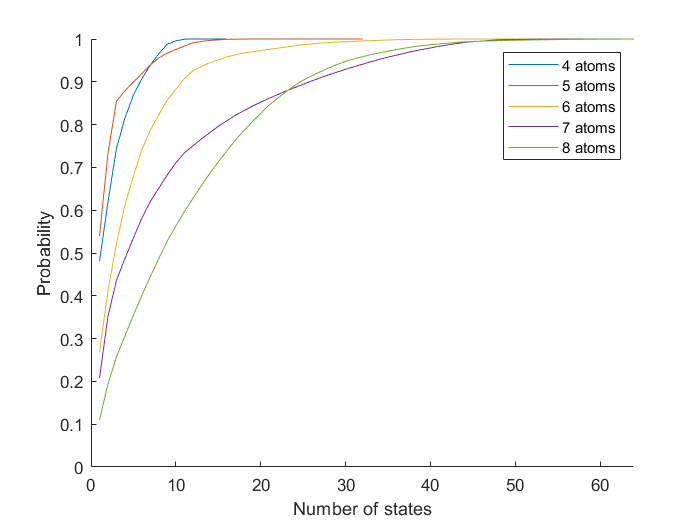}
}\qquad
\subfloat[]{
  \includegraphics[width=0.46\textwidth]{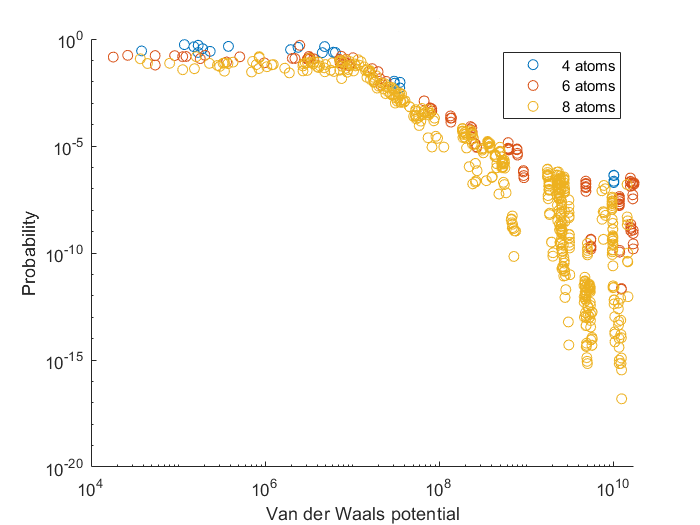}
}
\caption{Figure (a) shows the cumulative probability of states in the solution of the Schrödinger equation in the jamming limit without using state truncation for $\Omega = \SI{5.8}{\mega\hertz} $, $C_6 = \SI{50}{\giga\hertz \micro\meter^6}$, and $\rho = \SI{0.031}{\micro\meter^{-2}}$. The states are in order of decreasing probability. Figure (b) shows the probability of a state in the solution of the Schrödinger equation in the jamming limit without using state truncation over its Van der Waals potential.}
\label{fig:probstate}
\end{figure}

We also observe that the probability a state occurs is negatively correlated with its contribution to the Hamiltonian due to the Van der Waals interaction. Figure \ref{fig:probstate}(b) shows the probability a state occurs as a function of its Van der Waals potential. Clearly the probability diminishes for larger values of the Van der Waals potential. For our simulation we have truncated all states with a Van der Waals potential higher than $10^8$ resulting in a truncation of states with a probability of less than $10^{-4}$ of occurring. To verify we computed the maximum absolute difference between the mean number of excitations over time for a small truncated system and the statistics obtained without truncation and observed only fractional differences in all cases, confirming that the truncation does not influence our observations in any significant way.

After truncating the state space, each state is ordered lexicographically in the truncated state space such that it can equivalently be represented by an index $\pi(s) \in \mathbbm{N}$. We then compute the Hamiltonian matrix $H$. An element in the matrix $H_{i, j}$ is computed by determining the interaction energy of the states $s, s' \in S$ corresponding to the indices $i, j$. To compute the interaction energy the input parameters and the position matrix $r$ are used as stated in the main text. The Hamiltonian matrix results in the Schrödinger equation which we solve numerically as an ordinary system of differential equations.

\end{document}